\documentclass[reqno,10pt,final]{amsart}
\usepackage[left=2.8cm,right=2.8cm,top=2.8cm,bottom=2.8cm]{geometry}
\usepackage[utf8]{inputenc}
\usepackage[english]{babel}
\usepackage[T1]{fontenc}
\usepackage{lmodern}
\usepackage{amsmath}
\usepackage{amssymb,amsthm,amsfonts}
\usepackage{mathrsfs,mathtools,braket,esint}
\usepackage{dsfont,mathscinet,xfrac}
\usepackage{amsmath, amsfonts, amsthm, amssymb}
\usepackage{bm}

\usepackage{xargs}
\usepackage[pdftex,dvipsnames]{xcolor}
\usepackage[colorinlistoftodos,prependcaption,
,obeyFinal]{todonotes}

\usepackage[colorlinks=true, pdfstartview=FitV, linkcolor=blue, 
            citecolor=blue, urlcolor=blue]{hyperref}

\newtheorem{theorem}{Theorem}
\numberwithin{equation}{section}
\numberwithin{theorem}{section}

\newtheorem{lemma}[theorem]{Lemma}

\newtheorem{corollary}[theorem]{Corollary}
\newtheorem{definition}{Definition}
\newtheorem{remark}{Remark}
\numberwithin{remark}{section}

\renewcommand{\d}{\mathrm{d}}
\newcommand{\e}{\mathrm{e}}
\newcommand{\B}{\mathrm{B}}

\newcommand{\III}{\mathbb{I}}
\newcommand{\NNN}{\mathbb{N}}
\newcommand{\RRR}{\mathbb{R}}

\newcommand{\CCCC}{\mathcal{C}}
\newcommand{\DDDD}{\mathcal{D}}
\newcommand{\EEEE}{\mathcal{E}}
\newcommand{\FFFF}{\mathcal{F}}
\newcommand{\GGGG}{\mathcal{G}}
\newcommand{\HHHH}{\mathcal{H}}
\newcommand{\LLLL}{\mathcal{L}}

\newcommand{\QQQQ}{\mathcal{Q}}
\newcommand{\RRRR}{\mathcal{R}}

\newcommand{\YYYY}{\mathcal{Y}}

\newcommand{\BV}{\mathrm{BV}}
\newcommand{\forae}{\forall_{\text{a.e.}}}
\newcommand{\wto}{\rightharpoonup}

\DeclareMathOperator{\cof}{Cof}
\DeclareMathOperator{\diss}{Diss}
\DeclareMathOperator{\var}{Var}

\DeclareMathOperator*{\arglocmin}{ArgLocMin}
\DeclareMathOperator{\dom}{Dom}

\newcommand\EEE{\color{black}}
\newcommand{\mk}{\color{black}}
\newcommand{\ed}{\color{black}}
\begin{document}

\title{Separately Global Solutions to Rate-Independent Processes in Large-Strain Inelasticity}
\author[E. Davoli]{Elisa Davoli}
\address[E. Davoli]{
		Institute of Analysis and Scientific Computing,
		TU Wien,
		Wiedner Hauptstrasse 8-10, 1040 Vienna, Austria.
		Email: \href{mailto:elisa.davoli@tuwien.ac.at}{\tt elisa.davoli@tuwien.ac.at}.
}
\author[M. Kru\v{z}\'ik]{Martin Kru\v{z}\'ik}
\address[M. Kru\v{z}\'ik]{Czech Academy of Sciences, 
Institute of Information Theory and Automation, 
Pod Vod\'arenskou V\v{e}\v{z}\'i 4, 182 00 Prague, Czechia.
		Email: \href{mailto: kruzik@utia.cas.cz}{\tt kruzik@utia.cas.cz}.
}
\author[P. Pelech]{Petr Pelech}
\address[P. Pelech]{mathematical Institute, Charles University, Sokolovsk\'{a} 83, 186 00 Prague, Czechia.
		Email: \href{mailto:pelech@karlin.mff.cuni.cz}{\tt pelech@karlin.mff.cuni.cz}.
}
\maketitle

\begin{abstract}
In this paper, we introduce the notion of separately global solutions for large-strain rate-independent systems,
and we provide an existence result for a model describing bulk damage.
Our analysis covers non-convex energies blowing up for extreme compressions,
yields solutions excluding interpenetration of matter,
and allows to handle nonlinear couplings of the deformation and the internal variable featuring both Eulerian and Lagrangian terms.
In particular, \mk motivated by the theory developed in \cite{Roubicek2015} in the small strain setting,
and for separately convex energies we provide a solution concept suitable for large strain inelasticity. \EEE
\end{abstract}

\section*{Introduction}Rate independent systems (RIS) are characterized by \mk the \EEE lack of any internal time length scale:
rescaling the input of the system in time leads to the very same rescaling of its solution.
In continuum mechanics,
rate-independent models represent a~reasonable \ed ansatz whenever \EEE
 the external conditions change slowly enough
so that the system can always reach its equilibrium.
This applies if inertial, viscous, and thermal effects are neglected.
Rate independent systems have proven to be useful in modeling of hysteresis,
phase transitions in solids, elastoplasticity, damage,
or fracture in small and large strain regimes.
We refer to \cite{MieRou15RIST} for a thorough overview 
of various results and applications.

\ed In this paper we propose a new notion of local solutions for RIS describing large-strain inelastic phenomena, and we present an application to the setting of bulk-damage materials. Before introducing the definition of \emph{separately global solutions}, we briefly review the main mathematical features of RIS, as well as the most widely adopted solution concepts. \EEE

\ed A RIS is described by means of a \EEE
generalized gradient system,
shortly denoted by $(\QQQQ,\EEEE,\RRRR)$,
defined by the state space $\QQQQ$, the
energy functional $\EEEE:[0,T] \times \QQQQ \to (-\infty,+\infty]$
and a~dissipation potential $\RRRR: \mathcal{T}\QQQQ \to [0,+\infty]$,
where $[0,T]$ is the time interval
and $\mathcal{T}\QQQQ$ denotes the tangent bundle of $\QQQQ$.
Solutions to generalized gradient systems satisfy the so-called Biot's inclusion
\begin{align}\label{biot}
    0 \in \partial_{\dot q} \mathcal{R}(q(t),\dot q(t))
    +
    \partial_q \mathcal E(t,q(t)),
    \quad t \in [0,T],
\end{align}
where $\partial$ denotes a (generalized) subdifferential;
see Section \ref{sec:example} for a~specific choice
of~$(\QQQQ,\EEEE,\RRRR)$.
The dissipation potential $\RRRR$ is usually non-negative,
convex in its second argument,
and satisfies $\RRRR(q,0) = 0$ for all $q \in \QQQQ$.
Equation \eqref{biot} represents a~force balance
where the elastic forces $\partial_q \EEEE$ are in equilibrium
with static-friction forces $\partial_{\dot{q}} \RRRR$,
and where both of them are, in~general, multi-valued.
The formulation fully relies
on the theory of generalized standard materials;
see \cite{HalNgu75MSG,Fremond2002}.
The rate-independence is expressed
by the~positive one-homogeneity of $\RRRR(q,\cdot)$,
leading to a~zero homogeneity of the subdifferential
$\partial_{\dot{q}} \RRRR$,
\ed and to a time-scale invariance of the system. \EEE

 This property of the dissipation potential \ed causes \EEE the system \ed to become \EEE somewhat degenerate
and determines the non-smooth nature of rate-independent
processes (RIP).
In particular, when the energy $\EEEE$ is not convex,
solutions may develop jumps,
making the strong derivative $\dot{q}$ ill-defined:
a~reformulation of the equation \eqref{biot} is therefore necessary.
Observing that the positive
1-homogeneity of $\mathcal{R}$ rewrites as
\begin{align}
    \label{eq:1HomChar}
    \mathcal R(q,\dot q)= \langle \eta, \dot q\rangle
    \quad \text{whenever} \quad
    \eta \in\partial_{\dot q}\mathcal{R}(q,\dot q),
\end{align}
with $\langle\cdot,\cdot\rangle$ denoting
the duality pairing between $\mathcal{T}\QQQQ$
{
\color{OliveGreen}
and $\mathcal{T}^*\QQQQ$
},
and using the definition of the subdifferential
and the characterization \eqref{eq:1HomChar},
we see that the Biot's equation \eqref{biot} is equivalent
to the following two  conditions
\begin{subequations}
\label{eq:weak}
\begin{align}
    \label{localstability}
    \text{local stability:}
    \quad
    0
    &\in
    \partial_{\dot q} \mathcal{R}(q(t),0) + \partial_q \EEEE(t,q(t)), \\
    \text{power balance:}
    \quad
    0
    &=
    \label{eq:PowBal}
    \langle \partial_q \mathcal{E}(t,q(t)),\dot{q}(t) \rangle
    +
    \mathcal{R}(q(t),\dot{q}(t)),
\end{align}
\end{subequations}
we refer to \cite{MieRou15RIST} for the detailed derivation.
Note that the local stability is a~purely static condition,
saying that the static-friction forces have to be strong enough
to balance the elastic forces.
The second, scalar equation relates the power of the change of state
with the dissipation rate $\RRRR$.
It is a remarkable feature of RIS that one purely static condition
together with a~single scalar equation is enough to characterize
their evolution.
Integration of the power balance over $[0,T]$ then yields the energy equality
\begin{align} \label{e-equality}
    \EEEE(T,q(T)) + \int_0^T \RRRR(q(t),\dot{q}(t)) \, \d t
    =
    \mathcal{E}(0,q(0)) + \int_0^T \partial_t \EEEE(t,q(t)) \, \d t,
\end{align}
where the total $\RRRR$-variation $\int_0^T \RRRR(q(t),\dot{q}(t)) \, \d t$
expresses the total amount of dissipated energy,
and $\int_0^T \partial_t \EEEE(t,q(t)) \, \d t$ the  work
of~the loading.
Note that here $\partial_t$ stands for the \emph{partial} derivative of $\EEEE$
with respect to its first argument.
The sought weaker reformulation, completely free of any time derivatives,
is then finally obtained by substituting
the total $\RRRR$-variation by \ed the \EEE equivalent formula
\begin{align} \label{eq:diss}
    \diss_\mathcal{R}(q;[0,T])
    :=
    \sup
    \left\{
        \sum_{j=1}^N \mathcal{R}(q(t_j) - q(t_{j-1}));
        N \in \NNN, 0 \leq t_0 \leq \ldots \leq t_N \leq T
    \right\},
\end{align}
\ed representing the $\RRRR$-variation of  a $BV$-function $q$ on the interval $[0,T]$, \EEE
which is valid when $\QQQQ$ is a vector space,
the tangent bundle simplifies to $\mathcal{T}\QQQQ = \QQQQ \times \QQQQ$,
and also only for dissipation potentials
$\RRRR$ not depending on the state $q$, \mk i.e. $\RRRR=\RRRR(\dot q)$. \EEE

Apart from the non-existence of the strong derivative,
the aforementioned non-smoothness of RIS has yet another,
peculiar side effect,
that is the existence of a variety of different notions of solutions
 which differ significantly in the occurrence of jumps.
As mentioned already in \cite{Roubicek2015},
the proper definition of the solution to the \ed RIS \EEE
is here a vital part of the modeling.

The most general notion, encompassing all others,
is that of {\it local solution},
first introduced by R. Toader and C. Zanini in \cite{Toader-Zanini}.
It is defined by the local stability \eqref{localstability}
and by an \emph{upper energy inequality},
i.e. ``$=$'' in \eqref{e-equality} is replaced by ``$\le$'',
stating that the energy in the system cannot increase.
It is exactly the energy inequality which makes
the existence of local solutions easier to prove,
however, at the same time, it represents its main drawback.
Since the equality in \eqref{e-equality} is lost,
the energetics of the system \mk is \EEE not fully under control.
In other words, the dissipative mechanisms,
and hence the underlying physics of the system,
are not specified entirely.
From \ed a \EEE mathematical point of view the indefiniteness of the dissipative mechanisms
leads to
a selectivity gap,
as there may exist even an uncountable family of local solutions
to a~given problem; see e.g. \cite{Roubicek2015}.
Also from \ed a \EEE computational point of view the mere
inequality is disadvantageous. \ed Indeed, \EEE when the solution is approximated numerically,
it is not known whether the energy decreases
due to unspecified physical dissipation
or rather \mk due to rounding errors and numerical effects. \EEE

A generally stronger concept of {\it energetic solutions}
was first introduced by A. Mielke and F. Theil \cite{MielkeTheil2004}
and then advocated
by many authors (see \cite{MieRou15RIST} and the references therein).
The local stability \eqref{localstability} is replaced
by the~\emph{global stability} condition
\begin{align}
    \label{globalstability}
    \forall \tilde{q} \in \QQQQ: \quad
    \EEEE(t,q(t)) \leq \EEEE(t,\tilde{q}) + \DDDD(q(t),\tilde{q}),
\end{align}
where $\DDDD(q_1,q_2)$ is the so-called dissipation distance,
expressing the amount of dissipated energy when the state
changes from $q_1$ to $q_2$.
In the case in which $\QQQQ$ is a~vector space
coinciding with its tangent space,
the dissipation distance is related to the dissipation potential
by $\DDDD(q_1,q_2) = \RRRR(q_2-q_1)$,
provided that $\RRRR$ is again state independent;
we refer to \cite[Subsec.3.2.2]{MieRou15RIST} for a thorough
discussion about the relation between the dissipation potential $\RRRR$
and the dissipation distance $\DDDD$.

Energetic solutions are very flexible and applicable to convex
as well as nonconvex problems.
In nonconvex problems, however,
they do not necessarily provide  proper predictions of the mechanical behavior of the RIS, as  jumps of the solution
in time appear ``too early''
when compared to physical experiments in several applications,
for example in models predicting damage and fracture.
In fact, energetic solutions jump immediately when there is enough energy available,
as can be seen directly from \eqref{globalstability},
representing hence, in a~sense, the worst case scenario.
For this reason, the global stability condition is sometimes
called in the literature the \emph{energetic criterion},
in contrast to the Biot's equation \eqref{biot} representing a~\emph{stress criterion}.
In other applications, e.g.,  \mk the modeling of \EEE shape memory alloys,
the energetic criterion may provide a~reasonable simplification of the problem
that makes its modeling feasible.
The global stability condition also implies the lower energy
inequality (i.e.\ the opposite to the one in the definition
of the local solution).
Energetic solutions therefore satisfy even energy \emph{equality}
\eqref{e-equality},
so that the only dissipation is due to
$\int_0^T \RRRR(q,\dot q) \, \d t$;
in particular, there is no extra energy dissipated on jumps.
We again refer to \cite{MieRou15RIST}
for further properties of energetic solutions,
and for their applications in materials science.

Completely on the other side  of the spectrum with respect to energetic solutions
lie the so-called {\it BV solutions},
introduced by A. Mielke, R. Rossi and G. Savar\'{e} in \cite{MRS09},
which `jump as late as possible'.
BV solutions also satisfy the energy equality,
but with an extra dissipation \ed accounting \EEE for time discontinuities. This is given by a~detailed resolution of the jump
during which the viscous dissipative mechanism,
not present in the rate-independent limit, is again activated.
This notion of solution stems from models considering also viscous dissipation
as a limit for vanishing viscosity.
However, as opposed to the vanishing viscosity solutions,
BV solutions are not defined as a~pointwise limit of solutions
to a~viscous perturbation of the original RIS.
Unlike for energetic solutions,
the so far developed existence theory for BV solutions requires
very strong assumptions on the data
that can be barely expected to hold in the engineering practice.
Apart from the analytical point of view,
also the numerical approximation of BV solutions represents
a~fairly challenging task.

The list of possible notions of solutions presented here
is by no means exhaustive and the reader is encouraged to consult for example
\cite{Mielke2011RIP,MieRou15RIST,RindlerSchwarzacher2017,RindlerSchwarzacher2019}
for a~more detailed survey of both the mentioned and unmentioned solution strategies for RIS, as well as for recent developments.

\vspace{2ex}

Given the overview above, we regard the concept of local solutions as a
compromise between physical requirements and mathematical restrictions,
taking into account the state of the art of contemporary mathematical and numerical analysis.
As they are the most general concept,
local solutions do not exclude solutions jumping `late enough'
and represent a~well defined object towards which the numerical solutions may converge,
the numerics for local solutions being also more understood today than for BV solutions;
see \cite{Knees-Negri}.
\ed An application \EEE of local solutions in continuum mechanics,
was presented in \cite{Roubicek2015} with a~special regard to damage and delamination.
In this setting, the state space is linear and exhibits a~product structure
$Q = U \times Z$,
where $U$ denotes the~vector space of displacements $u$, also called elastic variables,
and $Z$ is the vector space of general internal variables $z$, describing,
e.g. the damage field. \mk Consequently, we write $q=(u,z)$. \EEE
In \cite{Roubicek2015}, the author proved the existence of local solutions to a rich class of problems
in which the energy functional $\EEEE(t,\cdot,\cdot)$
is not convex \ed on the state space, but \EEE is separately convex in the last two variables, i.e.
\begin{align}
    \label{eq:LocSolStrong}
    \mathcal{E}(t,\cdot,z)&: U \to (-\infty,+\infty]
    \quad \text{is convex and} \\
    \mathcal{E}(t,u,\cdot)&: Z \to (-\infty,+\infty]
    \quad \text{is convex}
\end{align}
and the Biot's equation \eqref{biot} is replaced by the system
\begin{align*}
    0 &\in \partial_u \EEEE(t,u(t),z(t)), \\
    0 &\in \partial_{\dot z} \mathcal{R}(\dot z(t))
    +
    \partial_z \mathcal E(t,u(t),z(t)).
\end{align*}
Note that the dissipation potential is assumed to be state independent.
The reason for calling $u$ the elastic variable is that the dissipation potential
does not depend on its rate, i.e. its change never dissipates energy.
The three convexities of $\RRRR(\cdot)$, $\EEEE(t,\cdot,z)$, and $\EEEE(t,u,\cdot)$
make the strong formulation \eqref{eq:LocSolStrong} equivalent
to global stability in $u$ and to the so-called semi-stability in $z$
\begin{align}
    \label{eq:LocWeak}
    \forall \tilde{u} \in U&: \quad
    \EEEE(t,u(t),z(t)) \leq \EEEE(t,\tilde{u},z(t)), \\
    \forall \tilde{z} \in Z&: \quad
    \EEEE(t,u(t),z(t)) \leq \EEEE(t,u(t),\tilde{z}) + \RRRR(\tilde{z} - z(t)).
\end{align}
The convexity in $u$ is also important for controlling the jump behavior in $u$\ed. This \EEE
is not directly encoded by the dissipation as is instead the case for the internal variable $z$.
The two separately global conditions in \eqref{eq:LocWeak} are then supplemented
\ed by \EEE an energy inequality where the total dissipation is given by \eqref{eq:diss}.
It should be noted that  \cite{Roubicek2015} also deals
with one of the main drawbacks of local solutions,
i.e their weak selectivity behavior,
and suggests to impose an additional criterion the solutions have to satisfy,
giving rise to maximally-dissipative local solutions.
\ed We point out that \EEE proving existence in this narrower class of solutions
is still an open problem.

\vspace{2ex}

The main goal of this article is to generalize the concept
of local solutions also to energies that are not necessarily convex
in the elastic variable and hence extending significantly
its applicability by including inelastic processes at large strain.
As already mentioned, the advantage over energetic solutions
is that `too early' jumps can be avoided.
The novelty of our analysis is
that the physical requirements typical \ed of \EEE large-strain mechanics,
in particular the local and global invertibility of the deformations,
are met within our setting.
We also stress that our study encompasses energies depending
on Eulerian gradients of the Eulerian fields,
which after rewriting into Lagrangian coordinates
introduce a~non-linear coupling between
the deformations and the Lagrangian fields of the internal variables.
We refer to Section \ref{sec:example} for a~specific example.

\ed In order to describe our contribution, we need to introduce some minimal notation. \EEE As in \cite{Roubicek2015},
we suppose that the state space is endowed with a product structure $\QQQQ = \YYYY \times Z$,
where $\YYYY$ is the space of elastic variables, here represented by deformations $y$,
and $Z$ is the vector space of internal variables $z$.
For a~time interval $[0,T]$ we consider
the energy functional
$\mathcal{E} : [0,T] \times \YYYY \times Z \to (-\infty,+\infty]$
and the dissipation functional
$\mathcal{R}: X \to [0,+\infty]$,
which is for simplicity not state dependent.
For analytical reasons we suppose $X \supset Z$,
i.e. that the tangent space $\mathcal{TQ}$ 
is larger than the state space of internal variables.
Since we work within the large-deformations framework,
we do not impose any convexity assumption on the functional
\begin{align}
\label{eq:not-convex}
  \mathcal{E}(t,\cdot,z):\YYYY \to (-\infty,+\infty].
\end{align}
Nevertheless,
we require convexity of the energy with respect to the Hessian of the deformations\mk, i.e., we work in the framework of nonsimple hyperelastic materials; cf.~\cite{BaCuOl81NLWCVP,Toupin1962,Toupin1964} \EEE.
Regarding the internal parameter, instead, we work under the assumption that
\begin{align}
\label{eq:convex}
  \mathcal{E}(t,y,\cdot):Z \to (-\infty,+\infty]
  \quad
  \text{\emph{is} convex}.
\end{align}
Concerning the dissipation potential,
we assume that $\mathcal{R}$ is a~so-called gauge, i.e.
\begin{align}
    \label{eq:gauge}
  \mathcal{R}: X \to [0,+\infty]
  \text{ is convex},
  \quad
  \forall a \geq 0, \, \forall z \in \dom \mathcal{R}:
  \mathcal{R}(az) = a \mathcal{R}(z),
\end{align}
where
$\dom \mathcal{R} := \{z \in Z:\mathcal{R}(z) < +\infty\}$.
The positive degree-1 homogeneity of $\mathcal{R}$
also implies $\mathcal{R}(0) = 0$ and,
together with convexity,
 guarantees that the dissipation potential satisfies the triangle inequality
$\mathcal{R}(a+b)\leq \mathcal{R}(a) + \mathcal{R}(b)$ for every $a,b\in X$.
We point out that a broad class of problems (such as delamination and bulk damage)
can be formulated within the framework
described by \eqref{eq:not-convex}, \eqref{eq:convex}, and \eqref{eq:gauge}.

A strong \ed formulation for \EEE our notion of solution consists in finding $(y,z):[0,T] \to \YYYY \times Z$ satisfying
\begin{subequations}
\label{eq:ClassicalSystem}
\begin{align}
  \label{eq:ClassicalSystem-IC}
  z(0) &= z_0, \\
  \label{eq:ClassicalSystem-y}
  y(t) &\in \arglocmin_{y \in \YYYY} \mathcal{E}(t,y,z(t)), \\
  \label{eq:ClassicalSystem-z}
  0 &\in \partial \mathcal{R}(\dot{z}(t))
  +
  \partial_z \mathcal{E}(t,y(t),z(t)),
\end{align}
\end{subequations}
where $z_0$ is a given initial condition,
$\partial$ and $\partial_z$ denote  (partial) sub-differentials,
and $\arglocmin$ is the set of local minimizers of the energy.
The weak formulation \ed of \eqref{eq:ClassicalSystem} \EEE relies on the convexity of the dissipation potential
and of the energy with respect to the internal variable.
Following \cite{Roubicek2015},
we rewrite \eqref{eq:ClassicalSystem-z}
using the concept of \emph{semi-stability} for the internal parameter $z$.
The requirement for the elastic variable to be a local minimum
is replaced by global minimality,
both because of \ed the stability of this condition \EEE and to ease its mathematical treatment.
\mk Let  $\B([0,T];U)$ stand for the space of functions defined everywhere in $[0,T]$ with values in $U$ that are  bounded.
\ed Recalling \eqref{eq:diss} for a $BV$-function $z$ on the interval $[r,s]\subset [0,T]$, our \mk notion of solution to the system \eqref{eq:ClassicalSystem} reads as follows\mk ; cf.~\cite[Def.~2.1]{Roubicek2015}.\EEE

\begin{definition}[a.e.-Separately Global Solution]\label{def:sep-global}
  The mapping $(y,z): t \mapsto (y(t),z(t)) \in \QQQQ$
  with $y \in \B(I;\YYYY)$ and
  $z \in \B(I;Z) \cap \BV (I;X)$
  is called an a.e.-separately global solution if
  $t \mapsto \partial_t \EEEE(t,y(t),z(t))$ is integrable,
  and the following conditions are satisfied:
    \begin{align*}
      &&&z(0) = z_0, \\
      &\forae t \in I, \, \forall \tilde{y} \in \YYYY:
      &&
      \mathcal{E}(t,y(t),z(t)) \leq
      \mathcal{E}(t,\tilde{y},z(t)), \\
      &\forae t \in I, \, \forall \tilde{z} \in Z:
      &&
      \mathcal{E}(t,y(t),z(t))
      \leq
      \mathcal{E}(t,y(t),\tilde{z})
      +
      \mathcal{R}(\tilde{z}-z(t)), \\
      &\forae t_1, t_2 \in I, \, t_1 < t_2:
      &&
      \mathcal{E}(t_2,y(t_2),z(t_2))
      +
      \diss_\mathcal{R}(z;[t_1,t_2]) \\
      &&& \quad
      \leq
      \mathcal{E}(t_1,y(t_1),z(t_1))
      +
      \int_{t_1}^{t_2}
        \partial_t \EEEE (t,y(t),z(t)) \, \d t.
    \end{align*}
\end{definition}

We point out that this concept
of solution does not require any time-differentiability of the internal variable $z$.
As in \cite{Roubicek2015}, though, no measurability in time of the deformation and no absolute-continuity in time of the internal variable are a-priori guaranteed. 

Separate global energy minimization is beneficial
from the mathematical point of view
and unlike the mere first optimality conditions guarantees the necessary stability of the solution,
which is required by the underlying physics.
On the other hand,
in some situations it inevitably leads to completely flawed predictions;
especially when time-dependent surface loads
are considered.
In order to avoid the most pathological situations
we hence treat only time-dependent Dirichlet boundary conditions
(sometimes referred to as `hard devices'),
or at least their relaxation via a~penalty method,
sometimes also called `soft-devices'.

Leaving the precise assumptions and formulations
to Section \ref{sec:mat-set},
we provide here a simplified statement of \ed our \EEE main result, explain the crucial steps of its proof,
and highlight the main novelties.

\begin{theorem}
\label{thm:main}
Under suitable assumptions on the system $(\QQQQ,\EEEE,\RRRR)$,
and in the setting of bulk damage, the problem \eqref{eq:ClassicalSystem}
admits a~separately global solution in the sense of Definition \ref{def:sep-global}.
\end{theorem}

The proof strategy relies on a classical procedure:
after performing a Rothe-type time discretization,
\ed we identify time-discrete solutions to the system
in Definition \ref{def:sep-global}. We then
establish uniform estimates for the associated piecewise-constant interpolants, and we
select suitable convergent subsequences \ed as the time-step $\tau$ tends to zero.\EEE
\ed Using the convexity of the energy \ed with respect to \EEE the hessian of the deformations and the gradients of the internal variables,
\ed we then prove that the previously selected subsequences satisfy improved convergence properties. \EEE
This step is instrumental for showing \ed that the obtained limiting maps satisfy the stability,
semi-stability, and energy-inequality conditions in Definition \ref{def:sep-global}.

The most important novelty of our contribution consists in
extending the notion of local solutions to the large-strain setting for RIS,
as well as to energies $\mathcal{E}$ which are not convex with respect to the deformations,
and which depend on the Eulerian gradient of the internal variable, namely on the quantity
(rewritten in Lagrangian coordinates) $\nabla y^{-\top} \nabla z$.
One difficulty caused by the fact that the gradient of the internal variable only appears in this latter coupled term  \mk is the low exponent with which  \EEE  $\nabla z$ is integrable. \ed This, in turn, makes the study of    compactness
properties for the sequence of piecewise-constant interpolants of time-discrete solutions extremely delicate.\EEE

\ed As a result, although an explicit dependence of the energy on the hessian of deformations is not needed for proving the existence of the time-discrete solutions, thus allowing to include classical
polyconvex energies \mk without \EEE higher-order terms in the analysis at the discrete level, a control on the second derivatives of the deformations is necessary for the passage from time-discrete to time-continuous solutions.

An additional complication in the passage from time-discrete to time-continuous solutions \ed is due to the fact that \EEE in the time-discrete stability condition \ed a \EEE \emph{right}-continuous interpolant of the internal variable appears,
\ed whereas in \EEE the time-discrete semi-stability \ed a \EEE \emph{left}-continuous interpolant
 is present \ed (see \eqref{eq:RInterp} and \eqref{eq:LInterp} below for the precise definitions). \EEE
\ed The standard estimates for the error between the two interpolants
are not sufficient for establishing compactness, thus calling for suitable refined techniques, which are presented in Lemma \ref{lemma:ImprovedZ} below. \ed A crucial idea is to  regularize the energy on the time-discrete level by an additional term (see \eqref{eq:RegS}), \EEE
which improves the integrability of the internal variable,
but disappears as the discretization parameter tends to zero.

Apart from the problems arising from geometric nonlinearities due to large strains,
we also have to face the non-uniqueness of solutions caused by the lack of \mk  (strict) \EEE convexity
in the elastic variable.
\ed As in \cite{FraMie06ERCR},
this \EEE leads to solutions where the elastic variable \mk may be non-measurable \EEE in time.
\ed A further difficulty lies in the fact that local solutions are characterized by an
energy inequality \mk comparing the energetics \EEE at \emph{two time instants},
which makes the compactness analysis even more complicated. \ed This issue is addressed in Subsections \ref{subs:t-dep} and \ref{subs:better-y} by extracting suitable time-dependent sequences exhibiting enhanced convergence properties. \EEE

\ed We finally observe that, although our study is focused on bulk-damage materials, more general ansatz for the energies, in which the \ed contributions of the highest order terms of the elastic and
internal variable are decoupled, could also be~considered and will be the subject of forthcoming works. \EEE

The paper is organized as follows:
in Section \ref{sec:example} we present our concrete case study.
The full statement of the main result is postponed to Theorem \ref{thm:existence}
in Section \ref{sec:mat-set}, where we describe the precise mathematical setting.
Section \ref{sec:proof} is devoted to the proof of Theorem \ref{thm:existence}.
Eventually,
in Section \ref{sec:disc} we summarize our contributions and discuss their main implications.

\section{A specific case study}
\label{sec:example}
Before proceeding further, we specify the abstract framework of RIS to a~concrete example from continuum mechanics,
postponing the full mathematical details to Section \ref{sec:mat-set};
see also \cite{MieRou15RIST} for a rich list of applications \mk to the modeling of solids. \EEE
In the problem of bulk gradient damage the state space $\QQQQ = \YYYY \times Z$
consists of deformations $y : \Omega \to \RRR^3$, for
$\Omega \subset \RRR^3$ denoting the body in the reference configuration,
and of damage fields $z: \Omega \to [0,1]$,
where $z(x) = 1$ means that the material at the point $x \in \Omega$ is flawless,
while $z(x) = 0$ corresponds to a fully damaged material,
completely unable to sustain any stress (perhaps with the exception of a pure compression).

The energy functional then takes the form
\begin{align} \label{eq:Ebulk}
  \mathcal{E}(t,y,z)
  =
  \int_\Omega
    V(z, \nabla y, \nabla^2 y) + \phi \left( (\nabla y)^{-\top} \nabla z \right) \,
  \d x
  -
  \ell (t,y,z),
\end{align}
where $V: \RRR \times \RRR^{3 \times 3} \times \RRR^{3 \times 3 \times 3} \to [0,+\infty)$
is the stored energy density of a nonsimple material (see e.g. \cite{BenKru17WLSCIFA, KruRou19MMCMS},
\mk $\phi: \RRR^3 \to (0,+\infty)$ \EEE prevents formation of~microstructures
by penalizing the `crack length' (in gradient damage relaxed by a~smooth interface),
and $\ell: [0,T] \times \YYYY \times Z \to \RRR$ denotes a time-dependent loading.
All the functions can be in principle $x$-dependent,
including also materials not homogeneous in the reference configuration.

For bulk damage, \ed in the homogeneous setting \EEE the stored energy density has usually the form
\begin{align*}
  V(\nabla y, \nabla^2 y, z)
  =
  \gamma(z) W(\nabla y, \nabla^2 y),
\end{align*}
where
$W:\RRR^{3 \times 3} \times \RRR^{3 \times 3} \to [0,+\infty)$
is the stored energy density of the undamaged material,
and the function $\gamma:\RRR \to (0,+\infty)$ models incomplete bulk damage,
meaning that its values are bounded \ed away \EEE from zero by some small positive constant
and hence preserving the coercivity of the stored energy density.
A particular choice of $W$ is
\begin{align*}
    W(\nabla y, \nabla^2 y)
    &=
    \psi(\nabla y)
    +
    \Phi(\nabla^2 y) \\
    &=
    a|\nabla y|^p
    +
    b|\cof \nabla y|^q
    +
    c (\det \nabla y)^r
    -
    d \ln (\det \nabla y)^2
    +
    \frac{\epsilon}{p} |\nabla^2 y|^p.
\end{align*}
The positive coefficients $a$,$b$,$c$,$d$
yield arbitrary Lamé constants;
see \cite{Ciar88ME1} for $p=2$.
The last term serves as a~mathematical regularization,
the coefficient $\epsilon > 0$ being small enough not to influence the material response considerably.
Thanks to the \ed enhanced regularity provided \EEE by $\Phi$,
the $\psi$-term can in principle be non-convex,
including for example St. Venant--Kirchhoff materials
\begin{align*}
    \psi(F)
    =
    \frac{1}{8}
    \CCCC
    \left( F^\top F - \III \right)
    : \left( F^\top F - \III \right).
\end{align*}
\ed In the expression above, $\CCCC$ is a~fourth order elastic tensor,
which is easy to use in an engineering practice thanks to its direct
connection to small-strain mechanics. \mk Additionally, $\III$ is the identity matrix. \EEE
A very common choice of $\gamma$ is
\begin{align*}
    \gamma(z)
    =
    \begin{cases}
        z^2 + \epsilon, & \text{for } z \geq 0, \\
        \epsilon, & \text{for } z < 0.
    \end{cases}
\end{align*}

The function $\phi$ depends on the Eulerian gradient of the Eulerian field
$\zeta(y(x)) = z(x), x \in \Omega$, denoted by $\nabla^y \zeta$.
An example is given by
\begin{align*}
    \int_{\Omega}
        \phi \left( (\nabla y)^{-\top} \nabla z \right) \, \d x
    =
    \int_{\Omega}
        \frac{\rho_0}{\alpha}
        \left| (\nabla y)^{-\top} \nabla z \right|^\alpha\,\d x
    =
    \int_{\Omega^y}
        \frac{\rho}{\alpha}
        \left| \nabla^y \zeta \right|^\alpha \, \d \ed x^y,
\end{align*}
where $\alpha>0$, and $\rho_0(x)=\rho(y(x))\det \nabla y(x)$ encodes the relation between
the Lagrangian and Eulerian fields of density with respect to the volume
in the reference and deformed configuration, respectively.
This choice favors the onset of damage at points where the density $\rho$ is smaller.

An example of the loading functional $\ell$,
 typically non-linear for large deformations,
would be
\begin{align}
    \label{eq:Loading}
    \ell(t,y,z)
    &=
    \int_\Omega B(t) \cdot y \, \d x
    +
    \int_{\Gamma_N} S(t) \cdot y \, \d \mathcal{H}^2(x) \\
    \nonumber &-
    \frac{1}{2 \varepsilon}
    \int_{\Gamma_{\mathrm{d}}}
        |y - y_D(t)|^{2} \, \d \mathcal{H}^2
    -
    \int_\Omega
        \pi(t,y) \det \nabla y \, \d x,
\end{align}
where $B$ and $S$ stand for
bulk and surface loads, respectively,
the third term penalizes 
the mismatch between the deformation value
and the prescribed Dirichlet
boundary condition $y_D$ for $\varepsilon \to 0$,
and the~last term is a~potential
for~a~surface pressure load,
$\pi$ being the~pressure field in the~deformed configuration.
\begin{remark}[Hyper-loading]
    \label{rem:NonSimple}
    Second grade materials, whose stored energy density
    depends also on $\nabla^2 y$,
    may model various physical phenomena,
    for example \ed the \EEE flow of Korteweg fluids
    (depending on the Eulerian gradient of the Eulerian density field)\ed,
    the \EEE deformation of woven fabrics \cite{Kort01FPEM,IsolaSteigmann2014}\ed, phase transitions \cite{BallCrooks2011,BallMora-Corral2009,Silh88PTNB},  and multisymplectic field theory \cite{Kouranbaeva2000}; \EEE see also 
    and \cite{Toupin1962,Toupin1964,GreRiv64:MCM,BaCuOl81NLWCVP,
    PideriSeppecher1997,Mariano2007,DeScVi09GHLISGM,
    Forest2009,SeppecherAlibertIsola2011,Segev2017}
    for further works on~non-simple continua,
    the list definitely not being exhaustive.
    \ed We point out that for the applications described in this paper it is not necessary to incorporate the hyper-loading,
    an additional terms representing conservative forces,
    typically in~a~form of an~edge traction or the so-called couple-stress
    or double force acting on the boundary 
    (see \cite{DavFri20TWR,DavFri20TWL,FriKru17PNLV,FRiKru20DVKP,Podio-Guidugli1990,Mindlin1964}), for
    no such physical phenomena are expected to arise in our intended applications.
\end{remark}

The second constitutive quantity in the gradient system
is the dissipation potential,
for applications to the unidirectional bulk damage being
\begin{align}
    \label{eq:DissPotent}
    \RRRR(\dot{z})
    =
    \begin{cases}
        \displaystyle
        \int_\Omega G \rho_0 |\dot{z}(x)|\,\d x
        =
        \mk \int_{\Omega^y} G \rho(x^y) |\dot{\zeta}(x^y)|\,\d x^y, \EEE
        &
        \text{for } \dot{z} \leq 0 \text{ a.e. in } \Omega, \\
        +\infty,
        &
        \text{otherwise},
    \end{cases}
\end{align}
where $G$ is the so-called fracture toughness and
$z(x) = \zeta(y(x))$ are respectively
the Lagrangian and Eulerian damage fields,  \mk and, for $x^y=y(x)$, we denote    $\rho(x^y)=\rho_0/\det\nabla y(x)$,  $ x \in \Omega$  the actual density. \EEE
In this model, the material cannot heal, as $\dot{z} > 0$ is prohibited.
The unidirectionality can be relaxed by allowing small healing,
which may be useful from the analytical point of view;
see \cite{Roubicek2015}.

\begin{remark}[Delamination]
    A~particular \ed instance \EEE of \ed a \EEE loading functional $\ell$
    depending also on the internal variable $z$ \ed is \EEE provided by delamination,
    where the `elastic constants' in the boundary term in \eqref{eq:Loading}
    may depend on $z$, simulating then a~damageable adhesive.
    The dissipation potential in \eqref{eq:DissPotent} would include then
    the corresponding surface term.
    We refer to \cite{Roubicek2015} for a specific example in the small-strain regime.
\end{remark}

\section{Mathematical Setup}
\label{sec:mat-set}
In this section, we introduce the precise setting of our work
and present the key mathematical ingredients for our proof.
In the whole paper we suppose \ed that \EEE the body's reference configuration
$\Omega \subset \RRR^3$ is a~bounded Lipschitz domain. \ed With a slight abuse of notation, we will sometime denote the $W^{m,p}$ and $L^p$-norms simply by $\|\cdot\|_{m,p}$ and $\|\cdot\|_p$ without specifying the domain and target space, whenever this will be clear from the context. \EEE

\subsection{State Space}
\label{subs:state}
The state space $\QQQQ$ is assumed to be endowed with the product structure
\begin{align}
    \label{eq:StateSpace}
    \QQQQ := \YYYY \times Z,
\end{align}
consisting of admissible deformations and internal variables,
given by
\begin{align}
\label{eq:def-Y}
    \YYYY
    :=
    \{
        y \in W^{2,p}(\Omega,\RRR^3)
        :
        (\det \nabla y)^{-1} \in L^s(\Omega),&\\
        \det \nabla y > 0 \text{ a.e. in } \Omega,
        \eqref{eq:C-N} \text{ holds}&
    \},
\end{align}
where \eqref{eq:C-N} is the Ciarlet--Nečas condition defined below, and
\begin{align}
\label{eq:def-Z}
    Z
    :=
    \{
    	z \in W^{1,\alpha}(\Omega) : 0 \leq z \leq 1 \text{ a.e. in } \Omega
    \},
\end{align}
for some $p,q,r \geq 1$, $s\geq0$, and $\alpha \geq 1$. 
Note that instead of time-dependent Dirichlet data,
representing a load by a hard device,
we consider their relaxation via a~penalty method,
 called soft-device.
While the relaxation is still sufficient for Poincaré-type inequalities on~$\YYYY$,
it does not \ed rely \EEE on extension theorems for locally invertible mappings,
which are up to now not available.
\ed In what follows, we will say that a sequence $\{y_n\}$ is uniformly bounded in $\YYYY$ whenever $\sup_{n\in \mathbb{N}}\{\|y_n\|_{2,p}+\|(\det\,\nabla y_n)^{-1}\|_s\}<+\infty$. \EEE

\subsection{Injectivity}
As outlined in the introduction,
the utmost physical requirement in large deformations
is that the material does not interpenetrate.
Our analysis \ed hinges upon \EEE  the classical Ciarlet--Nečas condition 
\begin{align} \label{eq:C-N}
  \int_\Omega \det\nabla y \,\d x
  \leq
  \LLLL^3(y(\Omega)),
\end{align}
where $\mathcal{L}^3$ denote the Lebesgue measure in $\RRR^3$. For other conditions ensuring injectivity
we refer to \cite[Sec. 6, Thm.2]{GiaquintaModicaSoucek1989} and to \cite{Ball81GISF};
however, these require the Dirichlet boundary datum to be prescribed
on the whole $\partial\Omega$ \ed (see also \cite{Kroemer2019}). \EEE
The original statement from \cite{CiaNec87ISCN}
states that whenever $p>3$ and $y\in W^{1,p}(\Omega;\RRR^3)$ is such that
$\det\nabla y>0$ almost everywhere in $\Omega$,
then the deformation $y$ is injective almost-everywhere in $\Omega$.
Apart from the Ciarlet--Nečas condition we use also the result
from \cite[Theorem~3.4]{HenKos14LMFD} by which under
the additional condition  
\begin{align} \label{eq:H-K}
  \frac{|\nabla y|^3}{\det\nabla y} \in L^{\delta}(\Omega),
\end{align}
 for some $\delta>2$, deformations satisfying \eqref{eq:C-N}
are also invertible \emph{everywhere} in $\Omega$.
The crucial observation is that \eqref{eq:H-K} then implies that $y$ is an open map.
Although the Ciarlet--Nečas condition \eqref{eq:C-N} is well suited
for proving the existence of minimizers by the direct method,
it is by no means trivial to incorporate \ed this \EEE non-local constraint numerically.
We refer e.g. to \cite{KroemerValdman2019} for some relaxations which are
more feasible for numerical computations, as well as for further generalizations.

\subsection{Boundedness of $\det \nabla y$ \ed away \EEE from zero}
Another important  ingredient for our proof,
helping to overcome the difficulties arising from the Eulerian gradient
of~the~damage field,
is the following corollary of the result from \cite{HeaKro09IWSS}.
For convenience of the reader
we provide an alternative proof in the Appendix \ref{app:anal}.

\begin{corollary}[Healey-Krömer \cite{HeaKro09IWSS} ]
	\label{corollary:HK}
    Let $p>3$, $s \ge \frac{3p}{p-3}$,    Then  for every $y\in \YYYY$ there exists $\varepsilon_y > 0$ s.t.
    \begin{align}
    \label{eq:det-bound}
        \det \nabla y \geq \varepsilon_y > 0
        \quad
        \text{in } \overline{\Omega}.
    \end{align}
    Moreover, if a~sequence $\{y_\tau\} \subset \YYYY$ is uniformly bounded in $W^{2,p}(\Omega;\RRR^3)$, and the sequence $\{(\det \nabla y_\tau)^{-1}\}$ is uniformly bounded in $L^s(\Omega)$, then the bound on the determinant in \eqref{eq:det-bound} is uniform.
\end{corollary}

The corollary implies that for the class $Z$ of internal variables, the Lagrangian gradient with respect to the set $\mathcal{Y}$ of admissible deformations is integrable if and only if
the Eulerian gradient has the same property. In other words, we have
\begin{align*}
    \nabla^y z
    := (\nabla y)^{-\top} \nabla z
    = \frac{(\cof \nabla y)}{\det \nabla y} \nabla z\in L^1(\Omega)\quad\text{iff}\quad \nabla z\in L^1(\Omega).
\end{align*}
This observation is crucial for proving the stability condition for the deformation $y$. 

Apart from the integrability of the \ed Eulerian \EEE gradient,
the Healey-Krömer lemma also allows us to prove continuity of elastic
energies that blow-up as the determinant of their argument converges to zero.
\begin{corollary}[Continuity of $W$ on $\YYYY$]
    \label{corollary:WConti}
    Let $\{y_k\}_{k\in\NNN} \subset \YYYY$ be such that
    $\{(\det \nabla y_k)^{-1}\}$ is uniformly bounded in $L^s(\Omega)$
    and $y_k \to y$ in $W^{2,p}(\Omega)$ as $k\to\infty$,
    where $p>3$ and $s \geq \frac{3p}{p-3}$.
    If the stored energy density
    $W:\RRR^{3 \times 3} \times \RRR^{3 \times 3 \times 3} \to [0,+\infty)$
    is continuous and \ed there exists a constant $C>0$ such that \EEE
    \begin{align*}
        |W(F,G)|
        \leq
        C(|F|^p + (\det F)^{-s} + |G|^p + 1),
    \end{align*}
    \ed for every $(F.G)\in \RRR^{3\times 3}\times \RRR^{3\times 3\times 3}$, \EEE then (for $k\to\infty$)
    \begin{align*}
        W(\nabla y_k, \nabla^2 y_k) \to W(\nabla y, \nabla^2 y)
        \quad
        \text{in } L^1(\Omega).
    \end{align*}
\end{corollary}
\begin{proof}
    Thanks to Corollary \ref{corollary:HK}
    the function $W$ has $p$-growth on the sequence \ed $\{y_\tau\}$. \EEE The thesis follows then from the standard argument for the continuity of the Nemytskii operator.
\end{proof}

\subsection{Strong Convergence Implied by Strict Convexity}
The well-known result for uniformly convex Banach spaces,
that weak convergence together with convergence of the norms implies strong convergence
(see e.g. \cite{FanGli58SGPS}),
has been generalized \ed in \EEE \cite{Visintin1984},
where, roughly speaking,
convergence of the norms is replaced by convergence of a strictly convex functional.
Since the case of multiple variables is not treated there,
we introduce here its slight generalization for integrands which are strictly convex
in the last variable.
The reader interested in the proof is referred to Appendix \ref{app:anal}.
\begin{lemma}[Strong convergence implied by the weak one]
    \label{lemma:Visintin}
    Let $\Omega \subset \RRR^3$ be a bounded Lipschitz domain.
    Assume that $C,c>0$ and
    $h:\Omega \times \RRR \times \RRR^{3\times3} \times \RRR^{3\times3\times3} \to \RRR$
    is a Carathéodory integrand such that for all
    $(x,z,F,G) \in \Omega \times \RRR \times \RRR^{3\times3} \times \RRR^{3\times3\times3}$
    and some $p,q \in (1,+\infty)$
    \begin{align}
        \label{h-growth}
        c(|z|^q + |F|^{p*} + |G|^p) - C
        \le
        h(x,z,F,G)
        \le
        C(1 + |z|^q + |F|^{p^*} + |G|^p)\ , 
    \end{align}
    where $h(x,z,F,\cdot)$ is strictly convex and
    \begin{align*}
        p^*:=
        \begin{cases}
            pn/(n-p) &\text{ if $1 < p < n$},\\
            \text{any number in } [1,+\infty) & \text{otherwise.}
        \end{cases}
    \end{align*}
    If $y_k\wto y$ in $ W^{2,p}(\Omega;\RRR^n)$, $z_k\to z$ in $L^q(\Omega)$, and 
    \begin{align}\label{limit11}
        \int_\Omega h(x,z_k,\nabla y_k,\nabla^2y_k)\, \d x
        \to
        \int_\Omega h(x,z,\nabla y,\nabla^2 y) \, \d x,
    \end{align}
    then $y_k\to y$ in $W^{2,p}(\Omega;\RRR^n)$.
\end{lemma}


\subsection{Free-Energy Assumptions}
\ed We consider the following \EEE ansatz
\ed for \EEE the energy functional $\EEEE: [0,T] \times \YYYY \times Z \to \mathbb{R}$
\begin{align}
    \label{eq:Energy}
    \mathcal{E}(t,y,z)
    =
		\int_\Omega
		    \gamma(z) W(\nabla y, \nabla^2 y)
		    +
		    \phi \left( (\nabla y)^{-\top} \nabla z \right) \,
		\d x
		-
		\ell (t,y,z),
\end{align}
We will assume that the functions
$\gamma$, $W$, $\phi$ and the loading $\ell$ satisfy the following
coercivity and growth conditions.
The function $\gamma:\RRR \to (0,+\infty)$ satisfies
\begin{align}
	\label{eq:gamma}
	\gamma \in C^1(\RRR) \text{ is positive and convex, and it is constant on } (-\infty,0].
\end{align}
The function $\phi$ is required to be strictly convex
and satisfy polynomial growth assumptions from above and below,
whereas $W$ is assumed to be convex in $\nabla^2 y$
and its growth conditions are compatible with the modeling of impenetrability
and blow-up of the elastic energy under strong compression.
Namely
\begin{align}
	\label{eq:phiCoer}
	c |u|^\alpha - C \leq \phi(u) \leq C (|u|^\alpha + 1),
\end{align}
for every $u\in \mathbb{R}^3$, and
\begin{align}
    \label{eq:WCoer}
	c(|F|^p + (\det F)^{-s} + |G|^p )
	\leq
	W(F,G) 
	\leq C(|F|^p + (\det F)^{-s} + |G|^p + 1),
\end{align}
for every $F\in \mathbb{R}^{3\times 3}$ and $G\in \mathbb{R}^{3\times 3\times 3}$, for some $p > 1$ and $s \geq 0$. Finally, we assume the nonlinear loading to fulfill the following coercivity assumption:
\begin{align}
	\label{eq:LCoer}
	-\ell(t,y,z)
	\geq
	c
	\left(
	    \left| \int_{\Gamma_D} y\,\d\mathcal{H}^2 \right|
	    -\|\nabla y\|^{\tilde{p}}_p
	    -\|\nabla^2 y\|^{\tilde{p}}_p
	    -\|(\det \nabla y)^{-1}\|^{\tilde{s}}_s
	    -\| (\nabla y)^{-\top}\! \nabla z\|^{\tilde{\alpha}}_\alpha
    \right)
    -C,
\end{align}
for every $(y,z)\in \mathcal{Q}$, for some $0 \leq \tilde{p} < p$,
$0 \leq \tilde{s} < s$, and $0 \leq \tilde{\alpha} < \alpha$.

In addition, for the loading $\ell$
we consider a slightly modified version of an assumption in \cite{Roubicek2015}.
We assume that the loading functional $\ell: \RRR \times \YYYY \times Z \to \RRR$
is such that for each $(\tilde{y},\tilde{z}) \in \QQQQ$ \ed both \EEE the map
\begin{align}
\label{eq:l1}
    (t,y,z) \mapsto \ell(t,\tilde{y},z) - \ell(t,y,z)
\end{align}
\ed and \EEE the map
\begin{align}
\label{eq:l2}
    (t,y,z,\tilde{z}) \mapsto \ell(t,y,\tilde{z}) - \ell(t,y,z)
\end{align}
have suitable (lower semi-)continuity properties;
see the hypotheses of Lemmas \ref{lemma:ImprovedZ}, \ref{lemma:ImprovedY},
\ref{lemma:SemiStab}, and \ref{lemma:Stab}.
Finally,
we assume that the loading enjoys the following structural properties
\begin{align}
    \label{eq:DomL}
	\dom \ell = [0,T] &\times \dom \ell(0,\cdot,\cdot), \\
        \label{eq:lw11}\forall (y,z) \in \QQQQ&:
            \quad \ell(\cdot,y,z) \in W^{1,1}(0,T), \\
    \label{eq:TDerEx}
        \forae t \in (0,T),\, \forall (y,z) \in \QQQQ&:
            \quad \partial_t \ell(\cdot,y,z) \text{ exists}.
\end{align}

Conditions \eqref{eq:gamma} - \eqref{eq:TDerEx} imply
the energetic control of~power,
described by the equations below
\begin{align} \label{eq:E2}
    \begin{aligned}
        &\dom \EEEE = [0,T] \times \dom \EEEE(0,\cdot,\cdot), \\
        &\forall (y,z) \in \QQQQ:
            \quad \EEEE(\cdot,y,z) \in W^{1,1}(0,T), \\
        &\forae t \in (0,T),\, \forall (y,z) \in \QQQQ:
            \quad \partial_t \EEEE(\cdot,y,z) \text{ exists}, \\
        &\exists C^0 \in \RRR\,\text{and } \lambda \in L^1(0,T)\text{ such that }
            \forae t \in (0,T)\text{ and } \forall (y,z) \in \QQQQ: \\
            & |\partial_t \EEEE(t,y,z)|
            \leq
            \lambda(t)(\EEEE(t,y,z) + C^0).
    \end{aligned}
\end{align}
The above properties, together with Gronwall's inequality, imply
\begin{align} \label{eq:E2G}
    \EEEE(t,y,z) + C^0
    \leq
    (\EEEE(s,y,z) + C^0) \e^{|\Lambda(t) - \Lambda(s)|},
    \quad
    \text{where }
    \Lambda(t) := \int_0^t \lambda(r) \, \d r,
\end{align}
which finally yields, when combined again with \eqref{eq:E2},
\begin{align} \label{eq:E2AC}
    |\partial_t \EEEE(t,y,z)|
    \leq
    \lambda(t) (\EEEE(s,y,z) + C^0) \e^{|\Lambda(t) - \Lambda(s)|},
\end{align}
for every $s,t\in [0,T]$.
 In particular, $t \mapsto \EEEE(t,y,z)$ is absolutely continuous for every $(y,z)\in \mathcal{Q}$.

\subsection{Assumptions on the Dissipation Potential}

We consider $\RRRR:L^1(\Omega) \to [0,+\infty)$
satisfying
\begin{align}
	\label{eq:RCoer}
	\inf_{v \neq 0} \frac{\RRRR(v)}{\|v\|_1} > 0,
\end{align}
and being lower-semicontinuous with respect to the \ed strong \EEE $L^1$-topology. 

Note that the convexity and positive 1-homogeneity together with $\RRRR(0) = 0$
imply that
\begin{align}
	\label{eq:R1}
		\forall z_1, z_2, z_3 \in Z:&
		\quad
		\RRRR(z_3 - z_1) \leq \RRRR(z_3 - z_2) + \RRRR(z_2 - z_1), \\
	\label{eq:R2}
		\forall v \in Z:&
		\quad
		\RRRR(v) = 0 \Longleftrightarrow v = 0,
\end{align}
i.e. $\RRRR$ defines a \ed norm. \EEE
\begin{remark}[\ed Unidirectional damage \EEE]
    Although in our assumptions the dissipation potential $\RRRR$ is finite everywhere,
    the setting of unidirectional damage can also be included in our analysis,
    see Remark \ref{rem:UniDam}.
\end{remark}

\subsection{Statement of the main result}

With the setting described in the previous subsections, our main result reads as follows.

\begin{theorem}[Existence Result]
    \label{thm:existence}
    Let $\Omega \subset \RRR^3$ be a~bounded Lipschitz domain. Let $\partial \Omega = \Gamma_D \cup \Gamma_N$ be a~measurable partition,
    with $\Gamma_D$ having a~positive two dimensional Hausdorff measure.
	Assume that the damage function $\gamma:\RRR \to (0,+\infty)$
	 is continuous and satisfies \eqref{eq:gamma}, and that
	the stored energy density
	$W: \RRR^{3\times3} \times \RRR^{3\times3\times3} \to [0,+\infty)$
	is continuous in its first argument, is convex in its second argument,
	and satisfies the coercivity condition \eqref{eq:WCoer} with $s\geq\frac{3p}{p-3}$ and $p > 3$.
	Assume also that $\phi : \RRR^{3 \times 3} \to \RRR$
	is strictly convex and fulfills the polynomial growth conditions in \eqref{eq:phiCoer}
	with $\alpha>3$. 
	Let the~dissipation potential $\RRRR: L^1(\Omega) \to [0,+\infty]$ be convex, 
	positively 1-homogeneous, lower semi-continuous
	with respect to the \ed strong \EEE $L^1(\Omega)$ topology,
	and satisfy \eqref{eq:RCoer}, \eqref{eq:R1}, and \eqref{eq:R2}.
    Assume that the coercivity condition \eqref{eq:LCoer} is satisfied,
    \ed that \eqref{eq:l1}--\eqref{eq:lw11} hold \EEE true, and that both
    the reduced power $-\partial_t \ell$ and the loading $\ell$
    are strongly continuous on \ed uniformly bounded \EEE sequences
    in~$\RRR \times \ed \YYYY \EEE \times W^{1,\alpha}(\Omega)$ \ed (see Subsection \ref{subs:state}). \EEE
    Then, problem \eqref{eq:ClassicalSystem}
    has a~separately global solution in~the sense of Definition \ref{def:sep-global},
    where in addition the deformation $y(t)$ is injective \emph{everywhere} in  $\Omega$ for all times $t \in [0,T]$.
\end{theorem}

\section{Proof of the Existence Theorem}
\label{sec:proof}
This section is devoted to the proof of Theorem \ref{thm:existence}.
The \ed argument \EEE follows a classical strategy,
namely performing a Rothe time discretization, proving
existence of discrete solutions, showing
compactness of piecewise-constant time-interpolants, improving the convergence,
passing to limit in the (semi-)stability condition,
and proving the energy inequality.
Since we regularize the problem on the discrete level,
we have to modify the initial condition
and hence the last step consists \ed in \EEE proving the convergence of the modified data to the original value $z_0$. \ed For convenience of the reader, the different steps of the proof are carried out into corresponding subsections. To be precise, in Subsection \ref{subs:rothe}, we describe the main properties of time-discrete solutions. Their existence is proved (under slightly more general assumptions) in Subsection \ref{subs:ex-d}. In Subsection \ref{subs:subs} we establish some first compactness properties for piecewise-constant time interpolants. These are  then improved in Subsections \ref{subs:better-z}--\ref{subs:better-y}. Eventually, Subsection \ref{subs:limit} is devoted to the study of limiting time-continuous solutions and is concluded by the proof of Theorem \ref{thm:existence}. \EEE
\subsection{Discrete Space Regularization}
As we will see later in~the~Subsection~\ref{subs:better-y},
the integrability of the \ed time-discrete \EEE solution needs to~be~improved.
We hence introduce the following energy regularization on~the time-discrete level
\begin{align}
    \label{eq:RegS}
    \mathcal{H}_\tau(z)
    :=
    \tau^{\kappa} \int_\Omega |\nabla z|^\beta \, \d x,
\end{align}
where $\kappa>0$ and $\beta > \alpha + \kappa$\ed, where $\alpha>3$ is the parameter in \eqref{eq:phiCoer}. \EEE
The role of this additional regularization
is to improve weak convergence of gradients for the time interpolants
associated to the damage field $z$.

On the other hand, the space regularization
prevents us from using $z \in Z$ as a~competitor in the discrete semi-stability.
We hence use its Lipschitz truncation instead; see Lemma \ref{lemma:Truncation}.
In our application,
we do not rely on the smallness of the set where the functions are truncated,
but only on the quantification of the Lipschitz constant.

\subsection{Rothe Time Discretization}
\label{subs:rothe}
Let $z^0 \in Z$ be the initial condition and let $y^0\in \mathcal{Y}$ be a minimizer of the \ed functional \EEE 
$y \mapsto \mathcal{E}(0, y, z^0)$.
\ed Due to the regularization \eqref{eq:RegS} \EEE we need
to construct our discrete solutions for a slightly modified initial condition.
We thus replace $z_0$ by its Lipschitz truncation $(z^0)_{\lambda(\tau)}$ \ed (see Lemma \ref{lemma:Truncation}), \EEE
where the truncation parameter $\lambda(\tau)$ depends on the time step
$\tau$ in such a way that
\begin{align}
    \label{eq:TruncZ0}
    \mathcal{H}_\tau((z^0)_{\lambda(\tau)})
    =
    \tau^\kappa \int_\Omega |\nabla (z^0)_{\lambda(\tau)}|^\beta \, \d x
    \leq C(\alpha,\Omega) \LLLL^3(\Omega) \tau^\kappa \lambda^\beta(\tau) \to 0,
\end{align}
where $C(\alpha,\Omega)$ is the constant from the bound
on $\|(z^0)_{\lambda(\tau)}\|_{1,\infty}$.
In other words, $\lambda(\tau)$ blows-up sufficiently slow. 
After performing a time discretization of $[0,T]$,
where $k \in \{1,\dots,N\}$, $N \in \NNN$ and $\tau = T/N$,
the regularized \ed time-discrete \EEE problem reads as follows.
\begin{subequations}
\label{eq:DiscreteSystem}
\begin{align}
    \label{eq:DiscreteSystem-IC-z}
    z^0_\tau &= (z^0)_{\lambda(\tau)}, \\
    \label{eq:DiscreteSystem-IC-y}
    y_\tau^0 &\text{ minimizes }
      y \mapsto \mathcal{E}(0, y, z_\tau^0) + \mathcal{H}_\tau(z_\tau^0) \\
    \label{eq:DiscreteSystem-y}
    y_\tau^k &\text{ minimizes }
      y \mapsto \mathcal{E}(k\tau, y, z_\tau^{k-1}) + \mathcal{H}_\tau(z_\tau^{k-1}) \\
    \label{eq:DiscreteSystem-z}
    z_\tau^k &\text{ minimizes }
    z \mapsto \mathcal{E}(k\tau, y_\tau^k, z) + \mathcal{H}_\tau(z)
        + \RRRR(z - z_\tau^{k-1}).
\end{align}
\end{subequations}
We postpone the study of the existence of solutions to \eqref{eq:DiscreteSystem}
to the next subsection and proceed by showing the properties of \ed time-discrete \EEE solutions. Recall \eqref{eq:E2G}.
The next lemma show that, provided these \ed time-discrete \EEE solutions exist, then they
satisfy standard uniform energy estimates.

\begin{lemma}[Properties of discrete solutions] \label{lemma:Discrete}
    Let $\dom \EEEE = [0,T] \times \YYYY \times Z$, assume that
    \eqref{eq:E2} holds, and let $z^0 \in Z$.
    Then any solution of the problem \eqref{eq:DiscreteSystem} satisfies
    for all $k \in \{1,\dots,N\}$ the discrete \ed stability, \EEE
    (semi)-stability, and energy inequality
    \begin{align}
    	\label{eq:DiscreteStab}
        & \forall \tilde{y}:
          && \mathcal{E}(k\tau,y_\tau^k,z_\tau^{k-1}) + \mathcal{H}_\tau(z_\tau^{k-1})
          \leq \mathcal{E}(k\tau,\tilde{y},z_\tau^{k-1}) + \mathcal{H}_\tau(z_\tau^{k-1}), \\
        \label{eq:DiscreteSemiStab}
        & \forall \tilde{z}:
          && \mathcal{E}(k\tau,y_\tau^k,z_\tau^k) + \mathcal{H}_\tau(z_\tau^k)
          \leq \mathcal{E}(k\tau,y_\tau^k,\tilde{z}) + \mathcal{H}_\tau(\tilde{z})
          +
          \RRRR(\tilde{z} - z_\tau^k), \\
        \nonumber
        &
        &&
        \mathcal{E}(k\tau,y_\tau^k,z_\tau^k)
        +
        \mathcal{H}_\tau(z_\tau^k)
        +
        \RRRR(z_\tau^k - z_\tau^{k-1}) \\
        \label{eq:DiscreteEIneq}
        &&& \quad
        \leq
        \mathcal{E}((k-1)\tau,y_\tau^{k-1},z_\tau^{k-1})
        +
        \mathcal{H}_\tau(z_\tau^{k-1})
        +
        \int_{(k-1)\tau}^{k\tau}
          \partial_t \EEEE(t,y_\tau^{k-1},z_\tau^{k-1}) \,
        \d t.
    \end{align}
    In addition, we have the uniform estimate
    \begin{align}
        &\nonumber \EEEE(k\tau,y_\tau^k,z_\tau^k) + \mathcal{H}_\tau(z_\tau^k) + C^0
        +
        \sum_{j=1}^k \RRRR(z_\tau^j - z_\tau^{j-1})\\
        &\label{eq:UniformEstDiscr}\quad\leq
        (\EEEE(0,y^0_\tau,z^0_\tau) + \mathcal{H}_\tau(z^0_\tau) + C^0)
        \e^{\Lambda(k\tau)}.
    \end{align}
\end{lemma}
\begin{proof}
    The discrete \ed stability and \EEE (semi)-stability \ed(\eqref{eq:DiscreteStab} and \eqref{eq:DiscreteSemiStab}) are direct consequences \EEE 
    of \eqref{eq:DiscreteSystem-y} and \eqref{eq:DiscreteSystem-z},
    respectively.
    Further, testing \eqref{eq:DiscreteSystem-z} by $z_\tau^{k-1}$, then \eqref{eq:DiscreteSystem-y} by $y_\tau^{k-1}$,
    and finally integrating in time (thanks to the absolute continuity of $\EEEE(\cdot,y,z)$ implied by \eqref{eq:E2})
    yield
    \begin{align*}
        &\EEEE(k\tau,y_\tau^k,z_\tau^k) + \mathcal{H}_\tau(z_\tau^k)
        + \RRRR(z_\tau^k - z_\tau^{k-1}) \\
        &\quad
        \leq
        \EEEE(k\tau,y_\tau^k,z_\tau^{k-1}) + \mathcal{H}_\tau(z_\tau^{k-1})
        \leq
        \EEEE(k\tau,y_\tau^{k-1},z_\tau^{k-1}) + \mathcal{H}_\tau(z_\tau^{k-1}) \\
        &\qquad=
        \EEEE((k-1)\tau,y_\tau^{k-1},z_\tau^{k-1}) + \mathcal{H}_\tau(z_\tau^{k-1})
        +
        \int_{(k-1)\tau}^{k\tau} \partial_t \EEEE(t,y_\tau^{k-1},z_\tau^{k-1}) \, \d t.
    \end{align*}
    The former estimate combined with \eqref{eq:E2G}
    implies
    (note that $\mathcal{H}_\tau$ does not depend on $t$ and hence we may apply
    the inequality also for $(\EEEE + \mathcal{H}_\tau)$)
    \begin{align*}
        &\EEEE(k\tau,y_\tau^k,z_\tau^k) + \mathcal{H}_\tau(z_\tau^k) + C^0
        + \RRRR(z_\tau^k - z_\tau^{k-1})\\
        &\qquad
        \leq
        \EEEE(k\tau,y_\tau^k,z_\tau^{k-1}) + \mathcal{H}_\tau(z_\tau^{k-1})+C^0 \\
        &\qquad
        \leq
        (\EEEE((k-1)\tau,y_\tau^{k-1},z_\tau^{k-1}) + \mathcal{H}_\tau(z_\tau^{k-1}) + C^0)
        \e^{\Lambda(k\tau) - \Lambda((k-1)\tau)}.
    \end{align*}
    Since $\RRRR(z_\tau^k - z_\tau^{k-1}) \geq 0$
    and for $k \in \{0,\dots,N\}$ we have
    $\EEEE(k\tau,y_\tau^k,z_\tau^k) + \mathcal{H}_\tau(z_\tau^k) < +\infty$,
    the discrete energy inequality follows arguing by induction;
    see e.g. \cite{FraMie06ERCR} or \cite[Thm 2.1.5]{MieRou15RIST}
    for details.
\end{proof}

The discrete stability for arbitrary time $t \in [0,T]$ will be proven in terms of piecewise-constant time interpolants. We collect the main definitions below.
\noindent
For $(k-1)\tau \leq t < k\tau$ we define
the right-continuous piecewise-constant interpolants
\begin{align}
    \label{eq:RInterp}
    \underline{z}_\tau(t) &:= z_\tau^{k-1},&\underline{y}_\tau(t) &:= y_\tau^{k-1},
\end{align}
while the left-continuous piecewise-constant interpolants are defined
for $(k-1)\tau < t \leq k\tau$ as
\begin{align}
    \label{eq:LInterp}
    \bar{z}_\tau(t) &:= z_\tau^k,
    &
    \bar{y}_\tau(t) &:= y_\tau^k
    &
    \bar{\mathcal{E}}_\tau(t,y,z) &:= \mathcal{E}(k\tau,y,z).
\end{align}
Note that the interpolants are related by
\begin{align} \label{eq:UnderlineBarRelNTK}
    \underline{y}_\tau(t+\tau)
    =
    \bar{y}_\tau(t),
    \quad
    \text{for every } t \in I,\,t\neq k\tau\text{ for }k\in \mathbb{N} ,
\end{align}
they coincide at the nodes, i.e.
\begin{align*}
    \underline{y}_\tau(t) = \bar{y}_\tau(t)
    \quad
    \text{for every } t = k \tau \in I, k \in \NNN,
\end{align*}
and they differ there only by their left/right continuity.
The same holds for $\underline{z}_\tau$ and $\bar{z}_\tau$.
For every $t\in [0,T]$, for notational convenience, we introduce the quantity
\begin{align*}
    \underline{\theta}_\tau(t)
    :=
    \partial_t \EEEE(t,\underline{y}_\tau(t),\underline{z}_\tau(t)),
\end{align*}
which will arise in~the energy inequality,
and its analogue
\begin{align*}
    \bar{\theta}_\tau(t)
    :=
    \partial_t \EEEE(t,\bar{y}_\tau(t),\bar{z}_\tau(t)),
\end{align*}
depending on~the~left-continuous interpolants.
Eventually, we define the pointwise limit
\begin{align*}
	\bar{\theta}^{\text{sup}}(t)
	:=
	\limsup_{\tau \to 0} \bar{\theta}_\tau(t).
\end{align*}
In the next lemma we show that these interpolants satisfy the standard discrete \ed stability, \EEE semi-stability, energy inequality
and uniform estimates.
\begin{lemma}[Properties of piecewise-constant interpolants]
	\label{lemma:InterpSol}
    Let the hypotheses of Lemma \ref{lemma:Discrete} and \eqref{eq:R1} hold.
    Then the interpolants satisfy
    for all $t \in [0,T]$ and all $0 \leq t_1 < t_2 \leq T$
    of~the~form $t_1 = k_1 \tau$ and $t_2 = k_2 \tau$ with $k_1, k_2 \in \NNN$
    \begin{align}
        \label{eq:InterpIC}
        &&& \bar{z}_\tau(0) = \underline{z}_\tau(0) = (z^0)_{\lambda(\tau)}, \\
        \label{eq:InterpStab}
        &\forall \tilde{y}\in \mathcal{Y}:
          && \bar{\mathcal{E}}_\tau(t,\bar{y}_\tau(t),\underline{z}_\tau(t))
          \leq \bar{\mathcal{E}}_\tau(t,\tilde{y},\underline{z}_\tau(t)), \\
        \label{eq:InterpSemiStab}
        & \forall \tilde{z}\in Z:
          && \bar{\mathcal{E}}_\tau(t,\bar{y}_\tau(t),\bar{z}_\tau(t))
            + \mathcal{H}_\tau(\bar{z}_\tau(t))
          \leq \bar{\mathcal{E}}_\tau(t,\bar{y}_\tau(t),\tilde{z})
            + \mathcal{H}_\tau(\tilde{z})
          +
          \RRRR(\tilde{z} - \bar{z}_\tau(t)), \\
        \label{eq:InterpEIneq}
        &&&\bar{\mathcal{E}}_\tau(t_2,\bar{y}_\tau(t_2),\bar{z}_\tau(t_2))
            + \mathcal{H}_\tau(\bar{z}_\tau(t_2))
        +
        \diss_{\RRRR}(\bar{z}_\tau; [t_1,t_2]) \\
        \nonumber
        &&& \quad
        \leq
        \bar{\mathcal{E}}_\tau(t_1,\bar{y}_\tau(t_1),\bar{z}_\tau(t_1))
        +
        \mathcal{H}_\tau(\bar{z}_\tau(t_1))
        +
        \int_{t_1}^{t_2}
          \partial_t \EEEE(t,\underline{y}_\tau(t),\underline{z}_\tau(t)) \,
        \d t.
    \end{align}
    Moreover, for every $t \in [0,T]$
    \emph{both interpolants},
    denoted here by $(\hat{y}_\tau,\hat{z}_\tau)$,
    satisfy the uniform estimates
    \begin{align}
    	&\EEEE(t,\hat{y}_\tau(t),\hat{z}_\tau(t))
    	+
    	\mathcal{H}_\tau(\hat{z}_\tau(t))
    	+
    	C^0
    	+
    	\diss_{\RRRR}(\hat{z}_\tau; [0,t]) \nonumber \\
    	\nonumber
    	&\quad
    	\leq
    	\e^{\Lambda(t)}(\EEEE(0,y^0,z_\tau^0) + \mathcal{H}_\tau(z_\tau^0) + C^0).
    	\label{eq:UniformEst}
    \end{align}
    Finally, if for every $y \in \YYYY$ the energy $\EEEE(0,y,\cdot)$
    is strongly continuous on $Z$, then
    \begin{align}
    	\e^{\Lambda(t)}(\EEEE(0,y^0,z_\tau^0) + \mathcal{H}_\tau(z_\tau^0) + C^0)
    	\to
    	\e^{\Lambda(t)}(\EEEE(0,y^0,z^0) + C^0)
    	< +\infty
    \end{align}
    as $\tau\to 0$, for every $t\in [0,T]$.
\end{lemma}
\begin{proof}
    The stability follows directly by Lemma \ref{lemma:Discrete}
    and by the definition of the interpolants,
    the regularizing term $\mathcal{H}_\tau(\underline{z}_\tau(t))$ is finite
    and cancels.
    For the semi-stability we further use the triangular inequality for $\RRRR$.
    \noindent
    Further, for the uniform estimate for arbitrary $t \in [0,T]$
    we need to exploit \eqref{eq:E2AC} once more;
    see \cite{FraMie06ERCR}.
        Finally, the convergence of the energies associated to the initial conditions follows by the strong convergence of the truncated
        initial conditions, with respect to which the energy is continuous,
        and by the choice of the truncation parameter \eqref{eq:TruncZ0}.
\end{proof}

\subsection{Existence of Discrete Solutions}
\label{subs:ex-d}
We proceed by proving the existence of the discrete solutions
defined in Subsection \ref{subs:rothe}.
We first show that existence of solutions is guaranteed in the setting described
in Subsection \ref{sec:mat-set}.
Then we prove existence of discrete solutions under relaxed
assumptions that do not require coercivity of the energy with respect to the hessian of the deformations.

\begin{lemma}[Existence of Discrete Solutions - \ed Non-simple materials\EEE]
\label{lemma:ex-2-grad}
    \label{lemma:DiscreteScndGrad}
	Let $\Omega \subset \RRR^3$ be a~bounded Lipschitz domain. Further, let $\partial \Omega = \Gamma_D \cup \Gamma_N$ be a~measurable partition,
    with $\Gamma_D$ having a~positive two dimensional Hausdorff measure.
	Assume that the damage function $\gamma:\RRR \to (0,+\infty)$
	 is continuous, and that
	the stored energy density
	$W: \RRR^{3\times3} \times \RRR^{3\times3\times3} \to [0,+\infty)$
	is continuous in its first argument, is convex in its second argument, and satisfies the coercivity condition \eqref{eq:WCoer},
	with $p > 3$. Assume also that
	$\phi : \RRR^{3 \times 3} \to \RRR$
	 is convex,
	and that the~dissipation potential $\RRRR: L^1(\Omega) \to [0,+\infty]$
	is lower semi-continuous with respect to the \ed strong \EEE $L^1$ topology.
    Finally, let $k\tau \in [0,T]$, $y^k_\tau \in \YYYY$, and $z^{k-1}_\tau \in Z$, and assume that the coercivity condition \eqref{eq:LCoer} is satisfied. 
    Under the further assumption that
    the loading $y \mapsto -\ell(k\tau,y,z^{k-1}_\tau)$
    is lower semi-continuous with respect to the weak topology
                on~$W^{2,p}(\Omega)$,
    and that $z \mapsto -\ell(k\tau,y^k_\tau,z)$
    is lower semi-continuous with respect to the weak topology  on~$W^{1,\alpha}(\Omega)$,
    then, the functional
    \begin{align*}
        \GGGG (z)
        &:=
        \EEEE(k\tau,y^k_\tau,z) + \mathcal{H}_\tau(z) + \RRRR(z - z^{k-1}_\tau) \\
        &=
        \int_\Omega
            \left(\gamma(z) W(\nabla y^k_\tau, \nabla^2 y^k_\tau)
            +
            \phi \left( (\nabla y^k_\tau)^{-\top} \nabla z \right)
            +
            \tau^\kappa |\nabla z|^\beta \,\right)
		\d x
        - \ell(k\tau,y^k_\tau,z) \\
        &\quad
        +
        \RRRR(z - z^{k-1}_\tau)
    \end{align*}
    has a minimizer $z^k_\tau \in Z$,
    and the functional
    \begin{align*}
        &\FFFF (y)
        :=
        \EEEE(k\tau,y,z^{k-1}_\tau) + \mathcal{H}_\tau(z^{k-1}_\tau) \\
        &=
        \int_\Omega
            \left(\gamma(z^{k-1}_\tau) W(\nabla y, \nabla^2 y)
            +
            \phi \left( (\nabla y)^{-\top} \nabla z^{k-1}_\tau \right)
            +
            \tau^\kappa |\nabla z^{k-1}|^\beta \,
		\right)\d x
        - \ell(k\tau,y,z^{k-1}_\tau)
    \end{align*}
    has a minimizer $y^k_\tau \in \YYYY$
    which is injective almost everywhere in~$\Omega$.
    If moreover $p > 6$ and $s > 2p / (p - 6)$,
    then the minimizer is injective \emph{everywhere} in $\Omega$.
\end{lemma}
\begin{proof}
Note that for $p > 3$ we have $(\nabla y^k_\tau)^{-\top} \in L^\infty$. The existence of a global minimizer follows thus by Fatou's lemma and by the direct method,
    while the injectivity is assured by \eqref{eq:C-N} and \eqref{eq:H-K}.
    
\end{proof}

The next two lemmas show existence of discrete solutions in the case
in which the elastic energy density satisfies weaker growth assumptions
than those in Lemma \ref{lemma:ex-2-grad}.
Namely, we suppose $W(\nabla y, \nabla^2 y) = W_{\mathrm{d}}(\nabla y)$,
i.e. no coercivity with respect to higher-order derivatives of the deformations is required.
We consider the state space $\QQQQ_{\mathrm{d}}$, defined as 
\begin{align*}
    \QQQQ_{\mathrm{d}}
    :=
    \{
        (y,z) \in \YYYY_{\mathrm{d}} \times Z_{\mathrm{d}}
        :
        \cof \nabla y \nabla z \in L^\alpha(\Omega)
    \}
    \subset
    \YYYY_{\mathrm{d}} \times Z_{\mathrm{d}},
\end{align*}
with
\begin{align*}
    \YYYY_{\mathrm{d}}
    :=
    \{
    	y \in W^{1,p}(\Omega,\RRR^3)
        :
        \cof \nabla y \in L^q(\Omega),
        \det \nabla y \in L^r(\Omega),&\\
        (\det \nabla y)^{-1} \in L^s(\Omega),
        \det \nabla y > 0 \text{ a.e. in } \Omega,
        \eqref{eq:C-N} \text{ holds}&
    \},
\end{align*}
and
\begin{align*}
    Z_{\mathrm{d}}
	:=
	\{
    	z \in W^{1,1}(\Omega) : 0 \leq z \leq 1 \text{ a.e. in } \Omega
	\},
\end{align*}
\ed for $p,q,r > 1$, $s \geq 0$, and $\alpha>3$, \EEE and we replace the energy functional $\EEEE$ by
\begin{align}
\label{eq:E-pc}
    \mathcal{E}_{\mathrm{d}}(t,y,z)
    :=
	\int_\Omega
	    \left(\gamma_{\mathrm{d}}(z) W_{\mathrm{d}}(\nabla y)
	    +
	    \phi_{\mathrm{d}} \left( \cof \nabla y \nabla z \right)\right) \,
	\d x
	-
	\ell_{\mathrm{d}} (t,y,z).
\end{align}
In definition \eqref{eq:E-pc} we suppose that
\begin{align}
	&\label{eq:gammaPC}
	\gamma_{\mathrm{d}} \in C(\RRR) \text{ and positive},
\end{align}
and that the energy density $W_{\mathrm{d}}:\RRR^{3\times 3} \to \RRR$
is polyconvex and satisfies the following coercivity assumptions
\begin{align}
    &\label{eq:WPolyCoer}
	W_{\mathrm{d}}(F) \geq c(|F|^p + |\cof F|^q + (\det F)^r + (\det F)^{-s} - 1),
\end{align}
for every $F\in \mathbb{R}^{3\times 3}$.
Further, we assume that
\begin{align}
	\label{eq:phiPolyCoer}
	\phi_{\mathrm{d}}(u)\geq c (|u|^\alpha - 1),
\end{align}
for every $u\in \RRR^3$, and that the nonlinear loading satisfies
\begin{align}
	\label{eq:LPCCoer}
	&-\ell_{\mathrm{d}}(t,y,z)
	\geq
	c
	\left(
	    \left| \int_{\Gamma_D} y\,\d\mathcal{H}^2 \right|
	    -\|\nabla y\|^{\tilde{p}}_p
	    -\|\cof \nabla y\|^{\tilde{q}}_q
	    -\|\det \nabla y\|^{\tilde{r}}_r\right.\\
	    &\nonumber\quad\left.-\|(\det \nabla y)^{-1}\|^{\tilde{s}}_s
	    -\|\cof \nabla y \nabla z\|^{\tilde{\alpha}}_\alpha
	    -1
    \right),
\end{align}
for every $(y,z)\in \QQQQ_{\mathrm{d}}$, for $0 \leq \tilde{p} < p$, $0 \leq \tilde{q} < q$, $0 \leq \tilde{r} < r$,
$0 \leq \tilde{s} < s$, and $0 \leq \tilde{\alpha} < \alpha$.

\ed In this setting, we seek solutions to the discrete problem
\begin{subequations}
\label{eq:DiscreteSystem-d}
\begin{align}
    \label{eq:DiscreteSystem-IC-z-d}
    z^0_\tau &= z^0, \\
    \label{eq:DiscreteSystem-IC-y-d}
    y_\tau^0 &\text{ minimizes }
      y \mapsto \mathcal{E}_{\rm d}(0, y, z_\tau^0) \\
    \label{eq:DiscreteSystem-y-d}
    y_\tau^k &\text{ minimizes }
      y \mapsto \mathcal{E}_{\rm d}(k\tau, y, z_\tau^{k-1}) \\
    \label{eq:DiscreteSystem-z-d}
    z_\tau^k &\text{ minimizes }
    z \mapsto \mathcal{E}_{\rm d}(k\tau, y_\tau^k, z) 
        + \RRRR(z - z_\tau^{k-1}).
\end{align}
\end{subequations}

We state below the existence results. \EEE We observe that the assumptions
of the \ed next two \EEE lemmas are fulfilled for problem \eqref{eq:DiscreteSystem-d},
provided \ed that \EEE the~initial condition is energetically stable.
For $k=1$, in fact, we can choose $\tilde{y} := y^0$ and $\tilde{z} := (z_0)_{\lambda(\tau)}$
because the domain of~\ed$\EEEE_{\rm d}$ \EEE
is independent of time and $\RRRR(0) = 0$. The same conclusion holds for every larger $k$ arguing by induction.

\begin{lemma}[Existence of Discrete Solutions - Deformations]
	\label{lemma:DiscreteDef}
	Let $\Omega \subset \RRR^3$ be a~bounded Lipschitz domain, let $\partial \Omega = \Gamma_D \cup \Gamma_N$ be a~measurable partition,
    with $\Gamma_D$ having a~positive two dimensional Hausdorff measure. Let
	$\gamma_{\mathrm{d}}:\RRR \to (0,+\infty)$ specify the incomplete damage, and let
	the stored energy density $W_{\mathrm{d}} : \RRR^{3\times3} \to \RRR$ be polyconvex and satisfy the coercivity assumption \eqref{eq:WPolyCoer}
	with $p \geq 2$, $q > p/(p-1)$, $r > 1$, $s \geq 1$.
	Let also $\phi_{\mathrm{d}} : \RRR^{3 \times 3} \to \RRR$ be convex, and satisfy the coercivity condition \eqref{eq:phiPolyCoer},
	where $\alpha$ is such that $1/p + 1/s + 1/\alpha \leq (q-1)/q$.
    Finally, let $k\tau \in [0,T]$ and $z^{k-1}_\tau \in Z_{\mathrm{d}}$ be such that
    there exists $\tilde{y} \in \YYYY_{\mathrm{d}}$ satisfying
	\begin{align}
		\label{eq:yFinE}
		(\tilde{y},z^{k-1}_\tau) \in \QQQQ_{\mathrm{d}},
		\qquad
		\EEEE_{\textrm{d}}(k\tau,\tilde{y},z^{k-1}_\tau) < +\infty.
	\end{align}
    Let the loading $y \mapsto -\ell_{\mathrm{d}}(k\tau,y,z^{k-1}_\tau)$
     be lower semi-continuous with respect to the weak product topology on~$W^{1,p}(\Omega) \times L^q(\Omega) \times L^r(\Omega)
        \ni (y,\cof \nabla y, \det \nabla y)$, and
        satisfy the coercivity condition \eqref{eq:LPCCoer}.

    Then, the functional
    \begin{align}
    \label{eq:fd}
        &\FFFF_{\mathrm{d}} (y)
        :=
        \EEEE_{\textrm{d}}(k\tau,y,z^{k-1}_\tau)
        =
        \int_\Omega
            \left(\gamma_{\mathrm{d}}(z^{k-1}_\tau) W_{\mathrm{d}}(\nabla y)
            +
            \phi_{\mathrm{d}} ( \cof \nabla y \nabla z^{k-1}_\tau )\right) \,
		\d x\\
        &\nonumber\qquad- \ell_{\mathrm{d}}(k\tau,y,z^{k-1}_\tau)
    \end{align}
    has a minimizer $y^k_\tau \in \YYYY_{\mathrm{d}}$ such that
    \begin{align*}
    	(y^k_\tau,z^{k-1}_\tau) \in \QQQQ_{\mathrm{d}},
    	\qquad
    	\EEEE_{\textrm{d}}(k\tau,y^k_\tau,z^{k-1}_\tau) < +\infty.
    \end{align*}
    If $p>3$,
    then the minimizer is injective almost everywhere in~$\Omega$.
    If $p > 6$ and $s > 2p / (p - 6)$,
    then the minimizer is injective \emph{everywhere} in $\Omega$.
\end{lemma}
\begin{proof}
    Since $\gamma_{\mathrm{d}}(z^{k-1}_\tau) > 0$, by the growth assumptions in \eqref{eq:LPCCoer} we need to address only the coercivity and  lower-semicontinuity of the $\phi_{\mathrm{d}}$ term. The remaining part of the proof follows
    from standard results
    on~polyconvex energies;
    see e.g. \cite{Ciar88ME1}.
	Let $y_n \rightharpoonup y^k_\tau$ be a minimizing sequence for $\FFFF_{\mathrm{d}}$ having finite energy. Note that its existence is guaranteed by~\eqref{eq:yFinE}. 
	By the coercivity of $\phi_{\mathrm{d}}$ we may further suppose that there exists $f\in L^\alpha(\Omega)$ such that, up to extracting not relabeled subsequences,
	\begin{align*}
	    \cof \nabla y_n \nabla z^{k-1}_\tau \rightharpoonup f
	    \quad
	    \text{in } L^\alpha(\Omega) \text{ as } n\to\infty.
	\end{align*}
	Since $\phi_{\mathrm{d}}$ is convex and finite,
	the second integrand in \eqref{eq:fd} is weakly lower semi-continuous with respect to the weak convergence in $L^\alpha$. Hence, it remains only to identify the limit $f \in L^\alpha(\Omega)$.
	Thanks to the estimate
	\begin{align}
	\label{eq:disc-grad-z}
		|\nabla z^{k-1}_\tau|
	    \leq
	    C
	    \left| \frac{\nabla y_n^\top}{\det \nabla y_n} \right|
	    |\cof \nabla y_n \nabla z^{k-1}_\tau|
	    \leq
	    \frac{C}{\det \nabla y_n} |\nabla y_n| |\cof \nabla y_n \nabla z^{k-1}_\tau|,
	\end{align}
	as well as to the choice of $p$, $s$, and $\alpha$, \ed and the coercivity of $W_{\mathrm{d}}$ and $\phi_{\mathrm{d}}$, \EEE we have $\nabla z^{k-1}_\tau \in L^{q'}(\Omega)$.
	Therefore, for every $\psi \in L^\infty(\Omega)$ we may choose $\varphi := \psi \otimes \nabla z^{k-1}_\tau$
	as a~test function for the weak convergence
	$\cof \nabla y_n \rightharpoonup \cof \nabla y^k_\tau$ in $L^q(\Omega)$. This yields
	\begin{align*}
		f = \cof \nabla y^k_\tau \nabla z^{k-1}_\tau \in L^\alpha(\Omega),
	\end{align*}
	and completes the proof of the lemma.
\end{proof}
\begin{remark}[\ed Assumptions on the exponents - $1$ \EEE]
\label{rk:exp1}
Let us explain how the assumptions on the exponents $p$, $q$, $s$, and $\alpha$ interact. 
Normally the existence of minimizers for polyconvex energies is known even for
$q \geq p/(p-1)$, i.e. $(q-1)/q \geq 1/p$.
However, due to the requirements on $\alpha$,
we have $1/p < 1/p + 1/s + 1/\alpha \leq (q-1)/q$.
Therefore, for smaller values of $s$ or $\alpha$, larger values of $q$ are needed.
\end{remark}

\begin{lemma}[Existence of Discrete Solutions - Damage variables]
	\label{lemma:DiscreteDamage}
	Let $\Omega \subset \RRR^3$ be a~bounded Lipschitz domain.
	Assume that the damage function $\gamma_{\mathrm{d}}:\RRR \to (0,+\infty)$ is continuous, and
	let the stored energy density $W_{\mathrm{d}} : \RRR^{3\times3} \to \RRR$ be bounded from below. 
	Let $\phi_{\mathrm{d}} : \RRR^{3 \times 3} \to \RRR$
	be convex, and satisfy the coercivity condition \eqref{eq:phiPolyCoer},
	with $\alpha$ such that $1/p + 1/s \leq (\alpha-1)/\alpha$.
	Eventually, let the~dissipation potential $\RRRR: L^1(\Omega) \to \ed [0,+\infty)\EEE$ be lower semi-continuous with respect to the \ed strong \EEE $L^1$-topology.
    Let $k\tau \in [0,T]$, $y^k_\tau \in \YYYY_{\mathrm{d}}$, and $z^{k-1}_\tau \in Z_{\mathrm{d}}$
    be such that there exists $\tilde{z} \in Z_{\mathrm{d}}$ satisfying
	\begin{align}
		\label{eq:zFinE}
		(y^k_\tau,\tilde{z}) \in \QQQQ_{\mathrm{d}},
		\qquad
		\EEEE_{\mathrm{d}}(k\tau,y^k_\tau,\tilde{z}) + \RRRR(\tilde{z} - z^{k-1}_\tau) < +\infty.
	\end{align}
    Let the loading $z \mapsto -\ell(k\tau,y^k_\tau,z)$
    be lower semi-continuous with respect to the weak $W^{1,1}$-topology, and
        satisfy the coercivity condition \eqref{eq:LPCCoer}.
    Then, the functional
    \begin{align*}
        \GGGG_{\mathrm{d}} (z)
        &:=
        \EEEE_{\mathrm{d}}(k\tau,y^k_\tau,z) + \RRRR(z - z^{k-1}_\tau) \\
        &=
        \int_\Omega
            \left(\gamma_{\mathrm{d}}(z) W_{\mathrm{d}}(\nabla y^k_\tau)
            +
            \phi_{\mathrm{d}} ( \cof \nabla y^k_\tau \nabla z )\right) \,
		\d x
        - \ell_{\mathrm{d}}(k\tau,y^k_\tau,z)
        +
        \RRRR(z - z^{k-1}_\tau)
    \end{align*}
    has a minimizer $z^k_\tau \in Z_{\mathrm{d}}$ such that
    \begin{align*}
    	(y^k_\tau,z^k_\tau) \in \QQQQ_{\mathrm{d}},
    	\qquad
    	\EEEE_{\mathrm{d}}(k\tau,y^k_\tau,z^k_\tau) + \RRRR(z^k_\tau - z^{k-1}_\tau) < +\infty.
    \end{align*}
\end{lemma}
\begin{proof}
	Let $\{z_n\}$ be a minimizing sequence for $\GGGG_{\rm d}$. We first observe that,
	thanks to \eqref{eq:zFinE}, $\{z_n\}$ has finite energy.
	By the continuity of $\gamma_{\mathrm{d}}$ and by the boundedness of $W_{\mathrm{d}}$ from below we have
	\begin{align*}
		\int_\Omega \gamma_{\mathrm{d}}(z_n) W_{\mathrm{d}}(\nabla y^k_\tau)\,\d x
		\geq
		-c \LLLL(\Omega) \max_{z \in [0,1]} \gamma_{\mathrm{d}}(z)
		> -\infty.
	\end{align*}
	Since $\RRRR \geq 0$,
	we infer from the coercivity assumptions on $\phi_{\mathrm{d}}$ and $\ell_{\mathrm{d}}$
	that $\{\cof \nabla y^k_\tau \nabla z_n\}$ is bounded in $L^\alpha(\Omega)$.
	Since $\{z_n\}$ is bounded in $L^\infty(\Omega)$ by the definition of $Z_{\mathrm{d}}$,
	up to subsequences we deduce
    \begin{align*}
        \cof \nabla y^k_\tau \nabla z_n &\rightharpoonup
        f
        \quad \text{in } L^\alpha(\Omega), \\
        z_n &\rightharpoonup^* z^k_\tau \quad \text{in } L^\infty(\Omega),
    \end{align*}
    for some $f \in L^\alpha(\Omega)$ and $z^k_\tau \in L^\infty(\Omega)$.

    To conclude, we need to characterize the limit $f$,
    and to show that $\nabla z^k_\tau\in L^1(\Omega)$.
    Thanks to the choice of $p$, $s$, and $\alpha$,
	the bound
	\begin{align*}
		|(\cof \nabla y^k_\tau)^{-T}|
		\leq
		\frac{|\nabla y^k_\tau|}{\det \nabla y^k_\tau}
	\end{align*}
    guarantees that \ed $(\cof \nabla y^k_\tau)^{-\top}  \in L^{\alpha'}(\Omega)$. \EEE
    Hence, we obtain
    \begin{align*}
        \nabla z_n
        \rightharpoonup
        (\cof \nabla y^k_\tau)^{-1} f
        \quad \text{in } L^1(\Omega).
    \end{align*}
    Since also $z_n \rightharpoonup z^k_\tau$ in $L^1(\Omega)$,
    by the definition of the distributional gradient we deduce
    \begin{align*}
        \nabla z^k_\tau
        =
    	(\cof \nabla y^k_\tau)^{-1} f \in L^1(\Omega),
    \end{align*}
    which in turn provides an identification of $f$.
    As a by-product we have also obtained, for a suitable (not relabeled) subsequence,
    that the following convergences hold true
    \begin{align*}
    	z_n \rightharpoonup z^k_\tau
    	\quad
    	\text{in } W^{1,1}(\Omega),
    	\qquad
    	z_n \to z^k_\tau
    	\quad
    	\text{in } L^{\hat{\alpha}}(\Omega),
    	\qquad
    	z_n \to z^k_\tau
    	\quad
    	\text{a.e. in } \Omega,
    \end{align*}
    for every $1 \leq \hat{\alpha} < 3/2$,
    which follows by the~compact embedding
    $W^{1,1}(\Omega) \Subset L^{\hat{\alpha}} (\Omega)$.

	To conclude the proof, it remains to show the lower semicontinuity
	of the functional $\GGGG_{\mathrm{d}}(z)$ with respect to these convergences.
	The lower semicontinuity of the first term is ensured by the continuity
	of $\gamma_{\mathrm{d}}$ and by Fatou's lemma.
	For the second term, we deduce it arguing as in Lemma \ref{lemma:DiscreteDef}.
	Finally, the loading $\ell$ and the dissipation potential $\RRRR$
	are lower semicontinuous with respect to the given convergences by assumption.
\end{proof}

\begin{remark}[Assumptions on the exponents - $2$]
\label{rk:exp2}
    Note that here small values of $p$ or $s$ imply large values of $\alpha$,
    and hence in turn, recalling Remark \ref{rk:exp1}, of $q$.
\end{remark}

\subsection{Selection of Subsequences}
\label{subs:subs}
In this subsection we analyze compactness properties of the sequences of time-interpolants
\eqref{eq:RInterp} and \eqref{eq:LInterp}
as the time-step $\tau$ converges to zero.
We recall that
	although the weak* limits of the left- and right-continuous interpolants
	coincide in~the~appropriate Bochner product space
	(see \cite[proof of Thm. 8.9]{RoubicekNPDE2013}
	which works also for piecewise-constant interpolants),
	the pointwise-in-time limits $\underline{y}(t)$ and $\bar{y}(t)$ of the $t$-dependent
	subsequences may differ in general.

\begin{lemma}[Compactness I]
    \label{lemma:Compactness1}
    Let the hypotheses of Lemmas  \ref{lemma:InterpSol}
	and \ref{lemma:DiscreteScndGrad} be satisfied.
    Let $\phi : \RRR^{3 \times 3} \to \RRR$
	satisfy the coercivity and growth condition \eqref{eq:phiCoer}
	with $\alpha > 1$.
  
    Finally,  assume that \eqref{eq:R1} and \eqref{eq:R2} hold.
    Then, we have the following uniform bounds
    \begin{align}
        \label{eq:n1}\|\bar{y}_\tau\|_{L^\infty((0,T),W^{2,p}(\Omega))} &\leq C,
        &\| (\det \nabla \bar{y}_\tau)^{-1} \|_{L^\infty((0,T),L^s(\Omega)))} &\leq C, \\
        \label{eq:n2}\|\bar{z}_\tau\|_{L^\infty((0,T),L^\infty(\Omega))} &\leq C,
        &\| (\nabla \bar{y}_\tau)^{-\top} \nabla \bar{z}_\tau\|_{L^\infty((0,T),L^\alpha(\Omega))} &\leq C, \\
        \label{eq:n3}\var_{L^1}(\bar{z}_\tau;[0,T]) &\leq C
        &\| \nabla \bar{z}_\tau\|_{L^\infty((0,T),L^\alpha(\Omega))} &\leq C, \\
    	\label{eq:n4}\|\bar{\theta}_\tau\|_{L^1(0,T)} &\leq C,
        &\|\underline{\theta}_\tau\|_{L^1(0,T)} &\leq C,
    \end{align}
    where $\{ |\bar{\theta}_\tau| \}$ and $\{ |\underline{\theta}_\tau| \}$ are
    equiintegrable.
    Up to a (not relabeled) subsequence,
    the following convergences of the piecewise\ed-constant \EEE interpolants associated to the internal variable hold true
    \begin{align}
    	\label{eq:HellyZ}
        \forall t \in [0,T]:&&
        \bar{z}_\tau(t) &\rightharpoonup^* z(t)
        &&\text{ and }
        &\underline{z}_\tau(t) &\rightharpoonup^* \underline{z}(t)
        &&\text{in } L^\infty(\Omega), \\
    	\label{eq:HellyNZ}
        &&
        \nabla \bar{z}_\tau(t) &\rightharpoonup \nabla z(t)
        &&\text{ and }
        &\nabla \underline{z}_\tau(t) &\rightharpoonup \nabla \underline{z}(t)
        &&\text{in } L^\alpha(\Omega).
    \end{align}
   The limiting maps  satisfy $z,\underline{z} \in \BV([0,T];L^1(\Omega)) \cap \B([0,T];Z)$,
    and there exists an at most countable set $J \subset [0,T]$ such that
    \begin{align}
        \label{eq:ZsEqual}
        \forall t \in [0,T] \setminus J:
        \quad
        z(t) = \underline{z}(t).
    \end{align}
   The dissipations associated to the left-continuous piecewise-constant interpolants of the internal variable fulfill
    \begin{align}
        \label{eq:HellyDiss}
        \forall t \in [0,T]&:
        \quad
        \diss_\RRRR(\bar{z}_\tau,[0,t]) \to \delta(t),
        \end{align}
        where $\delta\in \BV[0,T]$, and
        \begin{align}
        \label{eq:HellyDissLSC}
        \forall 0 \leq t_1 < t_2 \leq T&:
        \quad
        \diss_\RRRR(z,[t_1,t_2]) \leq \delta(t_2) - \delta(t_1).
    \end{align}
    The norms of the gradients of the interpolants of the internal variable satisfy
    \begin{align}
        \label{eq:NormZWeak}
        \| \nabla \bar{z}_\tau \|_{L^\alpha(\Omega;\RRR^{3\times 3})} \rightharpoonup^* \bar{f}
        \quad \text{and} \quad
        \| \nabla \underline{z}_\tau \|_{L^\alpha(\Omega;\RRR^{3\times 3})} \rightharpoonup^* \underline{f}
        \qquad
        \text{in } L^\infty(0,T),
    \end{align}
    for some $\bar{f}, \underline{f} \in L^\infty(0,T)$
    such that for almost every $t \in (0,T)$
    \begin{align}
        \label{eq:NormZUnderIneq}
		\underline{f}
		&\geq
		\underline{f}^{\text{inf}}(t)
		:=
		\liminf_{\tau \to 0} \| \nabla \underline{z}_\tau(t) \|_{L^\alpha(\Omega;\RRR^{3\times 3})}.
    \end{align}
    There exist $\bar{\theta},\underline{\theta}\in L^1(0,T)$ such that, up to a (not relabeled) subsequence, 
	\begin{align}
	    \label{eq:ThetaWeak}
		\underline{\theta}_\tau
		\rightharpoonup
		\underline{\theta}
		\quad \text{and} \quad
		\bar{\theta}_\tau
		\rightharpoonup
		\bar{\theta}
		\qquad
		\text{in } L^1(0,T).
	\end{align}
    Eventually, for the regularizing term we have the uniform bounds
    \begin{align}
        \label{eq:RegEBound}
        \mathcal{H}_\tau(\bar{z}_\tau(t)) \leq C,
        \quad
        \text{i.e.}
        \quad
        \int_\Omega |\nabla \underline{z}_\tau(t,x)|^\beta \, \d x
        \leq
        \frac{C}{\tau^\kappa}
    \end{align}
    for almost every $t\in (0,T)$.
\end{lemma}
\begin{proof}
    For convenience of the reader, we subdivide the proof into three steps.\\
    \textit{Step I: Uniform estimates.}
    Since $\gamma$ is positive and continuous,
    \ed \eqref{eq:n1}, \eqref{eq:n2}, and the first estimate in \eqref{eq:n3} \EEE are direct consequences of \eqref{eq:UniformEst},
    the coercivity of $\phi$ \eqref{eq:phiCoer},
    $W$ \eqref{eq:WCoer}, $\ell$ \eqref{eq:LCoer}, and $\RRRR$,
    and the definition of the space $Z$; see \eqref{eq:def-Z}.
    Note that
    \begin{align*}
        |\nabla \bar{z}_\tau| \leq C |\nabla \bar{y}_\tau^\top||(\nabla \bar{y}_\tau)^{-\top} \nabla \bar{z}_\tau|,
    \end{align*}
    and that thanks to \ed \eqref{eq:n2} \EEE and the embedding of $W^{2,p}(\Omega)$ into $C(\bar{\Omega})$ for $p>3$, 
    we have that $\{\nabla \bar{y}_\tau^\top\}$ is uniformly bounded in $L^\infty((0,T),L^\infty(\Omega))$. Using again the coercivity of $\phi$ \eqref{eq:phiCoer}, we thus obtain \ed the second estimate in \eqref{eq:n3}. \EEE

    \noindent
	Finally, the equiintegrability of $\bar{\theta}_\tau$ and $\underline{\theta}_\tau$, as well as the corresponding uniform bounds \eqref{eq:n4},
	follow by \eqref{eq:E2}, \eqref{eq:E2G} and \eqref{eq:UniformEst}.
	Indeed, for either of the interpolants, denoted by \ed the superscript ``$\hat{\phantom{g}}$'', \EEE we have
	\begin{align*}
	    |\hat{\theta}_\tau(t)|
		&=
		|\partial_t \EEEE(t,\hat{y}_\tau(t),\hat{z}_\tau(t))|
		\leq
		\lambda(t) (\EEEE(t,\hat{y}_\tau(t),\hat{z}_\tau(t)) + C^0) \\
		&\leq
		\lambda(t) \e^{\Lambda(t)} (\EEEE(0,y^0,z^0) + C^0 + 1),
	\end{align*}
	provided $\tau > 0$ is sufficiently small.

    \textit{Step II: Selection of sub-sequences.}
    By the convexity and  lower semicontinuity \ed of $\RRRR$ \EEE with respect to the strong $L^1$-convergence,
    we deduce that $\RRRR$ is lower semi-continuous \ed also \EEE in the \ed weak \EEE $L^1$- topology.
    \ed We observe that $\{\bar{z}_\tau(0)\}$ is uniformly bounded in $L^1(\Omega)$, thus by the first bound in \eqref{eq:n3} we can apply \cite[Lemma 7.2]{dalmaso-desimone-mora} to deduce that $$\bar{z}_\tau(t)\wto^*z(t)\quad\text{weakly* in }\mathcal{M}_b(\Omega)$$
    for every $t\in [0,T]$. \EEE

    Finally, functions $z:\Omega \to [0,1]$ form a~compact subset of $L^\infty(\Omega)$
    in the weak$^*$ topology
    and thanks to the uniform bounds proven in Step I, the sequence $\{\bar{z}_\tau\}$ is also weakly compact in $W^{1,\alpha}$ for almost every $t$.
    We hence apply the generalized Helly's selection principle
    from \cite{MieRou15RIST}
    to obtain a subsequence $\{\bar{z}_\tau\}$, not relabeled,
    satisfying \eqref{eq:HellyZ}, \eqref{eq:HellyNZ}, \eqref{eq:HellyDiss}, and \eqref{eq:HellyDissLSC}.

    \noindent
    Further, the convergence properties of $\{\underline{z}_\tau\}$, as well as the equality of $z$ and $\underline{z}$ almost everywhere in $[0,T]$
    follow from the fact that $\bar{z}_\tau - \underline{z}_\tau \to 0$
    in $L^1((0,T),L^1(\Omega))$; see \cite{Roubicek2015}.
    As both $z$ and $\underline{z}$ are in $\BV([0,T];L^1(\Omega))$
    they may differ only at discontinuity points which form an at most countable subset
    of $[0,T]$.

    \noindent
     \ed Eventually, \eqref{eq:NormZWeak} and \eqref{eq:ThetaWeak} are direct consequences of \eqref{eq:n3} and \eqref{eq:n4}, \EEE 
    while the inequality \eqref{eq:NormZUnderIneq} follows by Fatou's lemma. \ed
    Indeed, for every test function $\varphi\in L^1(0,T)$ with $\varphi\geq 0$ on $(0,T)$, from \eqref{eq:NormZWeak} and from the definition of $\underline{f}^{inf}$, we have
    $$\int_0^T \underline{f}^{inf}(t)\varphi(t)\,\d t\leq \int_0^T \underline{f}(t)\varphi(t)\,\d t.$$
    The thesis follows then by considering a sequence of test functions localizing around a point.
    \EEE

    \textit{Step III: Boundedness of the regularizing term.}
    The bounds \eqref{eq:RegEBound} follows directly from \eqref{eq:RegS} and the uniform estimate \eqref{eq:UniformEst}.
\end{proof}

\subsection{Improving the convergence of the internal variables}
\label{subs:better-z}
In this subsection, we improve the convergence \eqref{eq:HellyNZ} and show that for (almost) every $t \in [0,T]$ the piecewise-constant interpolants of the internal variables converge in the strong $W^{1,\alpha}$-topology. We point out that this property will be instrumental
both for passing to the limit in the right-hand side of the energy inequality,
and for identifying the reduced power.
\begin{lemma}[Improved Convergence in $z$]
    \label{lemma:ImprovedZ}
    Let the hypotheses of Lemma \ref{lemma:Compactness1} hold.
    Further, let $W: \RRR^{3 \times 3} \times \RRR^{3 \times 3 \times 3} \to [0,+\infty)$
    satisfy the coercivity condition \eqref{eq:WCoer},
   with $s \geq \frac{3p}{p-3}$. Let $\phi : \RRR^{3 \times 3} \to \RRR$
    be strictly convex and satisfy the coercivity and growth condition \eqref{eq:phiCoer},
   with $\alpha > 3$.
	Moreover, let the~dissipation potential $\RRRR$
	be continuous with respect to the \ed strong \EEE $L^1$-topology. 
    Finally, let the loading functional $\ell: \RRR \times \YYYY \times Z \to \RRR$
    be such that \ed for almost every $t\in [0,T]$ \EEE
    \begin{align*}
    	\forall \tilde{z}_{\tau} \to z(t) \text{ in } W^{1,\alpha}(\Omega):
    	\quad
        \liminf_{{\tau} \to 0}
            \left(
              \ell(k{\tau},\bar{y}_{\tau}(t),\tilde{z}_{\tau})
            - \ell(k{\tau},\bar{y}_{\tau}(t),\bar{z}_{\tau}(t))
            \right)
        \geq 0.
    \end{align*}
Then,
    \begin{align}
        \label{eq:zBarStrong}
        \forall t \in [0,T]:
        \quad
        \bar{z}_\tau(t) \to z(t)
        \quad
        \text{in } W^{1,\alpha}(\Omega),
    \end{align}
    and there exists a (not relabeled) subsequence such that
    \begin{align}
        \label{eq:zUnderStrong}
        \forae t \in (0,T):
        \quad
        \underline{z}_\tau(t) \to z(t)
        \quad
        \text{in } W^{1,\alpha}(\Omega).
    \end{align}
    Moreover, for the regularizing term we have
    \begin{align}
    \label{eq:reg-zero}
        \forall t \in [0,T]:
        \quad
        \mathcal{H}_\tau(\bar{z}_\tau(t)) \to 0.
    \end{align}
\end{lemma}
\begin{proof}
    We first prove the strong convergence \eqref{eq:zBarStrong}.
    Let $t \in [0,T]$ be given. \ed Recall the definition of $\beta$ below \eqref{eq:NormZWeak}. \EEE We observe that we cannot use the limit $z(t) \in Z$ as a~competitor
    in~the~discrete semi-stability \eqref{eq:InterpSemiStab},
    because it may not lie in $W^{1,\beta}(\Omega)$ \ed and hence the right-hand side of \eqref{eq:InterpSemiStab} would be infinite. \EEE We hence argue with its Lipschitz truncation $\tilde{z}_\tau := (z(t))_{\lambda(\tau)}$
    instead (see Lemma \ref{lemma:Truncation}).
    \ed We recall that \EEE the dependence of the truncation parameter $\lambda$ on the time step $\tau$
    is chosen to blow up sufficiently slow, i.e.
    \begin{align}
    \label{eq:blow-up-speed}
        \lim_{\tau \to 0} \tau^\kappa \lambda^\beta(\tau) = 0.
    \end{align}
    This gives
    \begin{align}
    \label{eq:new-bound}
    \bar{\EEEE}_\tau(t,\bar{y}_\tau(t),\bar{z}_\tau(t))+\mathcal{H}_\tau(\bar{z}_\tau(t))\leq \bar{\EEEE}_\tau(t,\bar{y}_\tau(t),\tilde{z}_\tau(t))+\mathcal{H}_\tau(\tilde{z}_\tau(t))+\RRRR(\tilde{z}_\tau(t)-\bar{z}_\tau(t)).     
    \end{align}

    Thanks to this uniform bound, and to the coercivity assumptions in the statement of the lemma, we extract a~$t$-dependent subsequence $\tau'(t) \to 0$
    (for ease of notation henceforth denoted by $\tau'$)
    for which
    \begin{align}
        \label{eq:y1Weak}
        \bar{y}_{\tau'}(t) &\rightharpoonup \xi(t)
        &&\text{ in } W^{2,p}(\Omega), \\
        \label{eq:Ny1NzWeak}
        (\nabla \bar{y}_{\tau'})^{-\top}\!(t) \nabla \bar{z}_{\tau'}(t)
        & \rightharpoonup (\nabla \xi)^{-\top}\!(t) \nabla z(t)
        &&\text{ in } L^\alpha(\Omega), \\
        \label{eq:RegEstrong}
        \mathcal{H}_{\tau'}(\bar{z}_{\tau'}(t)) &\to E(t) \geq 0, &&
    \end{align}
    for some $\xi(t) \in \YYYY$ \ed and $E(t)\in \RRR$. By \eqref{eq:phiCoer} and by the regularity of $\tilde{z}_{\tau'}$, the term $\int_\Omega \phi((\nabla \xi)^{-\top}(t) \nabla \tilde{z}_{\tau'}) \, \d x$ is finite for every $\tau'$. \EEE
    Subtracting \ed it \EEE
    from both sides of \eqref{eq:new-bound},
    we obtain after a rearrangement
    \begin{align}
    \label{eq:long-ineq}
        &\int_\Omega
          \phi \left((\nabla \bar{y}_{\tau'}(t))^{-\top} \nabla \bar{z}_{\tau'}(t) \right)
          -
          \phi \left((\nabla \xi(t)^{-\top} \nabla \tilde{z}_{\tau'} \right) \, \d x
        +
        \HHHH_{\tau'}(\bar{z}_{\tau'}(t)) \\
        &\nonumber \quad
        \leq
        \int_\Omega
          (\gamma(\tilde{z}_{\tau'}) - \gamma(\bar{z}_{\tau'}(t)))
          W(\nabla \bar{y}_{\tau'}(t), \nabla^2 \bar{y}_{\tau'}(t)) \, \d x \\
        &\nonumber \qquad
        + \RRRR(\tilde{z}_{\tau'} - \bar{z}_{\tau'}(t)) \\
        &\nonumber \qquad
        +
        \int_\Omega
          \phi \left((\nabla \bar{y}_{\tau'}(t))^{-\top} \nabla \tilde{z}_{\tau'} \right)
          -
          \phi \left((\nabla \xi(t))^{-\top} \nabla \tilde{z}_{\tau'} \right) \, \d x \\
        &\nonumber\qquad
        -
        \ell(k{\tau'},\bar{y}_{\tau'}(t),\tilde{z}_{\tau'})
        +
        \ell(k{\tau'},\bar{y}_{\tau'}(t),\bar{z}_{\tau'}(t))
        +
        \mathcal{H}_{\tau'}(\tilde{z}_{\tau'}),
    \end{align}
    where $k{\tau'} \searrow t$ as ${\tau'} \to 0$.

    Since $\phi$ is convex and finite, the first term on the left-hand side of \eqref{eq:long-ineq} is lower semi-continuous
    with respect to the convergence \eqref{eq:Ny1NzWeak}\ed. Recalling that
    \begin{equation}
        \label{eq:trun} \tilde{z}_\tau\to z(t)\quad\text{strongly in }W^{1,\alpha}(\Omega)
    \end{equation}
    (see Lemma \ref{lemma:Truncation}), by \eqref{eq:phiCoer} and \EEE the convergence
    of $\mathcal{H}_{\tau'}(\bar{z}_{\tau'}(t))$ in \eqref{eq:RegEstrong} we have
    \begin{align*}
        0
        \leq
        E(t)
        \leq
        \liminf_{{\tau'} \to 0}
        \int_\Omega
          \phi \left((\nabla \bar{y}_{\tau'})^{-\top}(t) \nabla \bar{z}_{\tau'}(t) \right)
          -
          \phi \left((\nabla \xi)^{-\top}\!(t) \nabla \tilde{z}_{\tau'} \right) \, \d x
        + 
        \mathcal{H}_{\tau'}(\bar{z}_{\tau'}(t)).
    \end{align*}
    Concerning the right-hand side of \eqref{eq:long-ineq},
    from \eqref{eq:y1Weak} and the growth assumptions on $W$ (see \eqref{eq:WCoer}), we deduce
    \begin{align*}
        &\int_\Omega
          (\gamma(\tilde{z}_{\tau'}) - \gamma(\bar{z}_{\tau'}(t)))
          W(\nabla \bar{y}_{\tau'}(t), \nabla^2 \bar{y}_{\tau'}(t)) \, \d x \\
        &\quad
        \leq
        \|\gamma(\tilde{z}_{\tau'}) - \gamma(\bar{z}_{\tau'}(t))\|_\infty
        \int_\Omega
          W(\nabla \bar{y}_{\tau'}(t), \nabla^2 \bar{y}_{\tau'}(t)) \, \d x \\
        &\quad
        \leq
        C
        \|\gamma(\tilde{z}_{\tau'}) - \gamma(\bar{z}_{\tau'}(t))\|_\infty \to 0.
    \end{align*}
    The latter convergence follows from
    the uniform continuity of $\gamma$ on the compact interval $[0,1]$
    and the uniform convergence of $\tilde{z}_{\tau'} - \bar{z}_{\tau'}(t)$ to zero.
    This, in turn, is achieved by the strong convergence of the Lipschitz truncations
    and the weak convergence \eqref{eq:HellyNZ} of $\bar{z}_{\tau'}(t)$
    in $W^{1,\alpha}(\Omega)$,
    both combined with the (compact) Sobolev embedding of $W^{1,\alpha}(\Omega)$ into $C(\bar{\Omega})$ for $\alpha > 3$.

    Concerning the third term on the right-hand side of \eqref{eq:long-ineq},
    the growth condition \eqref{eq:phiCoer} for $\phi$ implies continuity with respect to the following convergences
    \begin{align*}
        (\nabla \bar{y}_{\tau'})^{-\top}\!(t) \nabla \tilde{z}_{\tau'}
        &\to
        (\nabla \xi)^{-\top}\!(t) \nabla z(t)
        \quad
        \text{in } L^\alpha(\Omega), \\
        (\nabla \xi)^{-\top}\!(t) \nabla \tilde{z}_{\tau'}
        &\to
        (\nabla \xi)^{-\top}\!(t) \nabla z(t)
        \quad
        \text{in } L^\alpha(\Omega),
    \end{align*}
    which in turn hold thanks to \eqref{eq:trun}
    and  \eqref{eq:y1Weak}.
    By the Cramer's rule,
    \begin{align*}
        (\nabla \bar{y}_{\tau'})^{-\top}
        =
        \frac{\cof \nabla \bar{y}_{\tau'}}{\det \nabla \bar{y}_{\tau'}}.
    \end{align*}
    \ed On the other hand, by \eqref{eq:y1Weak},
    Corollary \eqref{corollary:HK}, and by the growth condition \eqref{eq:WCoer}, we deduce \EEE
    \begin{align*}
       \sup_{\tau>0} \left\| \frac{1}{\det \nabla \bar{y}_{\tau'}} \right\|_\infty \leq C.
    \end{align*}
    \ed Thus, by the compact embedding of $W^{2,p}(\Omega)$ into $C(\bar{\Omega})$ for $p > 3$, \EEE we obtain
    \begin{equation}
        \label{eq:Ny1Strong}
        (\nabla \bar{y}_{\tau'})^{-\top}\!(t)
        \to (\nabla \xi)^{-\top}\!(t)
        \quad\text{ in } L^\infty(\Omega).
\end{equation}
    \ed By \eqref{eq:trun} and \EEE
    by the assumptions on the loading we infer that
    \begin{align*}
        &\limsup_{{\tau'} \to 0}
            \left(
            - \ell(k{\tau'},\bar{y}_{\tau'}(t),\tilde{z}_{\tau'})
            + \ell(k{\tau'},\bar{y}_{\tau'}(t),\bar{z}_{\tau'}(t))
            \right)
        \leq 0.
    \end{align*}

    Finally, the regularization term on the right-hand side satisfies
    \begin{align}
    \label{eq:reg-zero1}
        \mathcal{H}_{\tau'}(\tilde{z}_{\tau'})
        =
        (\tau')^\kappa \int_\Omega |\nabla \tilde{z}_{\tau'}|^\beta \, \d x
        \leq
        C \LLLL^3(\Omega) (\tau')^\kappa \|\tilde{z}_{\tau'}\|_{1,\infty}^\beta
        \leq
        C \LLLL^3(\Omega) (\tau')^\kappa \lambda^\beta({\tau'})
        \to 0,
    \end{align}
    owing to \eqref{eq:blow-up-speed}.

    Altogether we have
    \begin{align*}
        \limsup_{{\tau'} \to 0}
        \int_\Omega
          \phi \left((\nabla \bar{y}_{\tau'})^{-\top}\!(t) \nabla \bar{z}_{\tau'}(t) \right)
          -
          \phi \left( (\nabla \xi)^{-\top}\!(t) \nabla z(t) \right) \, \d x
        \leq
        0,
    \end{align*}
    as well as $E(t) = 0$.
    Consequently, we deduce
    \begin{align*}
        \int_\Omega
          \phi \left( (\nabla \bar{y}_{\tau'})^{-\top}\!(t) \nabla \bar{z}_{\tau'}(t) \right) \, \d x
        \to
        \int_\Omega
          \phi \left( (\nabla \xi)^{-\top}\!(t) \nabla z(t) \right) \, \d x.
    \end{align*}
    By the strict convexity of $\phi$
    it even holds (cf. \cite{Visintin1984})
    \begin{align*}
        (\nabla \bar{y}_{\tau'})^{-\top}\!(t) \nabla \bar{z}_{\tau'}(t)
        \to
        (\nabla \xi)^{-\top}\!(t) \nabla z(t)
        \quad
        \text{in } L^\alpha(\Omega).
    \end{align*}
    \ed By \eqref{eq:y1Weak}, and by the compact embedding of $W^{1,p}(\Omega)$ into $C(\bar{\Omega})$, we find \EEE
    $(\nabla \bar{y}_{\tau'})^{\top}\!(t) \to (\nabla \xi)^{\top}\!(t)$ in $L^\infty(\Omega)$.
    \ed Thus, we conclude that \EEE
    \begin{align}
    \label{eq:conv-grad-z}
        \nabla \bar{z}_{\tau'}(t) \to \nabla z(t)
        \quad
        \text{in } L^{\alpha}(\Omega).
    \end{align}
   Since the same procedure applies to any subsequence $\tau' \to 0$, we deduce that
  the convergence holds for the full sequence $\{\bar{z}_\tau\}$.

    \ed We proceed by showing \EEE the strong convergence
    of $\{ \underline{z}_\tau(t) \}$ in $W^{1,\alpha}(\Omega)$.
    Since for $1 < \alpha < +\infty$ the space $L^\alpha(\Omega)$
    is uniformly elliptic,
    it suffices to show the convergence of the norms of the gradients.
    Using the already proven strong convergence \eqref{eq:zBarStrong}
    and Lebesgue's Theorem we infer \ed that
    for almost every $t \in [0,T]$ there holds $\bar{f}(t) = \| \nabla z(t) \|_\alpha$, where $\bar{f}$ is the map introduced in \eqref{eq:NormZWeak}.
    The weak convergence \eqref{eq:HellyNZ} implies
    \begin{align*}
        \forall t \in [0,T] \setminus J:
        \quad
        \| \nabla z(t) \|_\alpha
        \leq
        \liminf_{\tau \to 0} \| \nabla \underline{z}_{\tau}(t) \|_\alpha
        =
        \underline{f}^{\text{inf}}(t),
    \end{align*}
    where we used  \eqref{eq:ZsEqual}.
    Recalling \eqref{eq:NormZUnderIneq} we have
    for almost every $t \in (0,T)$
    \begin{align}
        \label{eq:fBarUnderIneq}
        \bar{f}(t)
        = \| \nabla z(t) \|_\alpha
        \leq \underline{f}^{\text{inf}}(t) \leq \underline{f}(t).
    \end{align}
    \ed To conclude the proof, we need \EEE to show the equality of $\bar{f}$ and $\underline{f}$ in $L^\infty(0,T)$.
    This is based on the relation \eqref{eq:UnderlineBarRelNTK}
    between the left- and right-continuous interpolants,
    holding almost everywhere in $(0,T)$.
    For every $\varphi \in L^1(0,T)$ we use the substitution
    $t = s + \tau$ and obtain
    \begin{align*}
        &\int_0^T \| \nabla \underline{z}_\tau(t) \|_\alpha \varphi(t) \, \d t
        =
        \int_0^\tau
            \| \nabla \underline{z}_\tau(t) \|_\alpha \varphi(t) \,
        \d t
        +
        \int_0^T
            \| \nabla \bar{z}_\tau(s) \|_\alpha
            \chi_{(0,T-\tau)}(s) \varphi(s+\tau) \,
        \d s \\
        & \quad
        \to
        \int_0^T \bar{f} \varphi(s) \, \d s,
    \end{align*}
    where the convergence follows by the absolute continuity of the Lebesgue integral
    and the strong continuity of translations in $L^1(0,T)$.
    Hence the inequalities in \eqref{eq:fBarUnderIneq} are in fact equalities
    and we apply \cite[Lemma 3.5]{FraMie06ERCR} to conclude that
    \begin{align*}
        \| \nabla \underline{z}_\tau \|_\alpha
        \to
        \| \nabla z \|_\alpha
        \quad
        \text{in } L^1(0,T).
    \end{align*}
    In particular, there is a~subsequence such that for almost every $t \in (0,T)$ we have
    $\| \nabla \underline{z}_\tau(t) \|_\alpha \to \| \nabla z(t) \|_\alpha$.
\end{proof}

\begin{remark}[Unidirectional Damage]
    \label{rem:UniDam}
    In the setting of unidirectional damage, for which
    $\dom \RRRR = \{v \in L^1(\Omega): v \leq 0 \text{ a.e. in } \Omega\}$\ed, instead of arguing with the Lipschitz truncations provided by Lemma \ref{lemma:Truncation},
    we can use a slight modification of a construction in
    \cite{MielkeRoubicek2006} and set
    \begin{align*}
        \tilde{z}_\tau :=
        \left(
            z(t))_{\lambda(\tau)}
            -\|\underline{z}_\tau(t) - (z(t))_{\lambda(\tau)}\|_\infty
        \right)^+
        \to
        z(t) \quad \text{in } W^{1,\alpha}(\Omega).
    \end{align*}
    \ed Since in our proof we need to have $\alpha>3$ ,
    \ed we cannot profit from the construction in \cite{MielkeThomas2010} which would have required a lower values of $\alpha$. \EEE
\end{remark}

\subsection{Selection of $t$-Dependent Subsequences}
\label{subs:t-dep}

Starting from the subsequence $\tau \to 0$ identified in Lemma \ref{lemma:ImprovedZ}, and for which
in particular \eqref{eq:zUnderStrong} holds,
we now select further $t$-dependent subsequences whose limits will define the solution
$y \in \B([0,T];\YYYY)$.
As in \cite{FraMie06ERCR}, this procedure allows to work only with subsequences which in the limit
`maximize the \ed work of the loading' \EEE.
\begin{lemma}[Compactness II - $t$-Dependent subsequences]
    \label{lemma:Compactness2}
    Let the hypotheses of Lemma \ref{lemma:ImprovedZ} be satisfied.
    Then for almost every $t \in (0,T)$
	\begin{align}
		\label{eq:ThetaIneq}
		\bar{\theta}(t)
		\leq
		\bar{\theta}^{\text{sup}}(t)
		:=
		\limsup_{\tau \to 0} \bar{\theta}_\tau(t),
	\end{align}
	and for all $t \in [0,T]$ there exists a $t$-dependent subsequence $\{\tau(t)\}$ such that
    \begin{align}
        \label{eq:ThetaPW}
        \bar{\theta}_{\tau(t)}(t) &\to \bar{\theta}^\text{sup}(t), && \\
        \label{eq:yWeak}
        \bar{y}_{\tau(t)}(t) &\rightharpoonup y(t)
        &&\text{ in } W^{2,p}(\Omega),
    \end{align}
    where $y \in \B([0,T];\YYYY)$. Consequently,
    \begin{align}
        \label{eq:NyStrong}
        (\nabla \bar{y}_{\tau(t)})^{-\top}\!(t)
        & \to (\nabla y)^{-\top}\!(t)
        &&\text{ in } L^\infty(\Omega), \\
        \label{eq:NyNzWeak}
        (\nabla \bar{y}_{\tau(t)})^{-\top}\!(t) \nabla \bar{z}_{\tau(t)}(t)
        & \to (\nabla y)^{-\top}\!(t) \nabla z(t)
        &&\text{ in } L^\alpha(\Omega).
    \end{align}
    Moreover, for every $t\in [0,T]$ there exists a~sequence
    $\{v_{\tau(t)}(t)\} \subset W^{1,\alpha}(\Omega)$ such that
    \begin{align}
        \label{eq:w1alpha}\|v_{\tau(t)}(t)\|_{1,\infty} &\leq \frac{C(\alpha,\Omega,t)}{\tau(t)}, \\
        \label{eq:meas-Mtau}\LLLL^3 (M_{\tau(t)}(t)) &\leq C(t) (\tau(t))^\alpha, \\
        \label{eq:equiint}\{|\nabla v_{\tau(t)}(t)|^\alpha\} & \text{ is equiintegrable},
    \end{align}
    where
    $M_{\tau(t)}(t) := \{ x \in \Omega : \underline{z}_{\tau(t)}(t) \neq v_{\tau(t)}(t)
        \text{ or } \nabla \underline{z}_{\tau(t)}(t) \neq \nabla v_{\tau(t)}(t) \}$.
\end{lemma}

\begin{proof}
    The convergence \eqref{eq:ThetaPW} can be achieved by a proper choice
    of the subsequence. Thanks to the uniform estimates proven in Lemma \ref{lemma:Compactness1}, a further selection procedure
    yields \eqref{eq:yWeak}.
    To prove that $y(t) \in \YYYY$ we first observe that the weak convergence in $W^{2,p}(\Omega)$ in \eqref{eq:yWeak}
    yields pointwise convergence of $\{\det \nabla \bar{y}_{\tau(t)}(t)\}$.
   Hence the uniform energy bounds in Lemma \ref{lemma:InterpSol}, together with Fatou's lemma, imply
    \begin{align*}
        +\infty >\liminf_{\tau\to 0} \int_\Omega \frac{\d x}{(\det \nabla \bar{y}_{\tau(t)}(t))^s}
        \geq
        \int_\Omega \frac{\d x}{(\det \nabla y(t))^s},
    \end{align*}
    i.e. $(\det \nabla y(t))^{-1} \in L^s(\Omega)$ and consequently
    $\det \nabla y(t) > 0$ a.e. in $\Omega$.
    Second, we recall that the Ciarlet-Nečas condition is stable under weak $W^{1,p}$
    convergence; c.f. \cite{CiaNec87ISCN}.

    The strong convergence of $\{(\nabla \bar{y}_{\tau(t)})^{-\top}(t)\}$ in $L^\infty(\Omega)$
    follows by the very same argument as in the proof of Lemma \ref{lemma:ImprovedZ}.
    Property \eqref{eq:NyNzWeak} follows  from \eqref{eq:zBarStrong}.
    The existence of the sequence $\{ v_{\tau(t)}(t) \}$ is a consequence of Lemma \ref{lemma:Decomposition}.
\end{proof}

\subsection{Improving the convergence of the elastic variables}
\label{subs:better-y}

As we argued in Lemma \ref{lemma:ImprovedZ} for the damage variables, we show that for every $t\in [0,T]$ we can extract further $t$-dependent subsequences satisfying stronger compactness properties than those in Lemma \ref{lemma:Compactness1}. Recall the limiting maps identified in Lemma \ref{lemma:Compactness1}.

\begin{lemma}[Improved convergence of the elastic energies]
	\label{lemma:ImprovedY}
    Let the hypotheses of Lemma \ref{lemma:ImprovedZ} hold, and for every $t\in [0,T]$ let $\{\tau(t)\}$ be the subsequence identified in Lemma \ref{lemma:Compactness2}.	
    Let the loading functional $\ell: \RRR \times \YYYY \times Z \to \RRR$
    be such that
        $\liminf_{\tau \to 0}
            \left(
              \ell(k\tau,y(t),\underline{z}_\tau(t))
            - \ell(k\tau,\bar{y}_\tau(t),\underline{z}_\tau(t))
            \right)
        \geq 0.$

    Then, for all $t \in [0,T]$ we have
    \begin{align}
    \label{eq:lim-en}
        \int_\Omega
          \gamma(\underline{z}_{\tau(t)}(t)) W(\nabla \bar{y}_{\tau(t)}(t), \nabla^2 \bar{y}_{\tau(t)}(t)) \, \d x
        \to
        \int_\Omega
          \gamma(\underline{z}(t)) W(\nabla y(t), \nabla^2 y(t)) \, \d x,
    \end{align}
    and
    \begin{align}
        \label{eq:yStrong}
        \bar{y}_{\tau(t)}(t) \to y(t)
        \quad \text{in } W^{2,p}(\Omega).
    \end{align}
    In particular, for every $t \in [0,T]$
    \begin{align}
        \label{eq:ElastConv}
        \int_\Omega
          \gamma(\bar{z}_{\tau(t)}(t)) W(\nabla \bar{y}_{\tau(t)}(t), \nabla^2 \bar{y}_{\tau(t)}(t)) \, \d x
        \to
        \int_\Omega
          \gamma(z(t)) W(\nabla y(t), \nabla^2 y(t)) \, \d x.
    \end{align}
\end{lemma}

\begin{proof}
    Let $t \in [0,T]$.
    For ease of notation throughout this proof we will simply write
    $\tau$ instead of $\tau(t)$.
    We point out, nevertheless, that all sequences will be $t$-dependent.
    By using the limit $y(t) \in \YYYY$ as a~competitor
    in~the~discrete stability \eqref{eq:InterpStab}
    and subtracting from both sides the term $\int_\Omega \gamma(\underline{z}(t)) W(\nabla y(t), \nabla^2 y(t))\,\d x$,
    we obtain after rearranging some terms
    \begin{align}
    \label{eq:improved-z}
        &\int_\Omega
          \left(\gamma(\underline{z}_\tau(t)) W(\nabla \bar{y}_\tau(t), \nabla^2 \bar{y}_\tau(t))
          -
          \gamma(\underline{z}(t)) W(\nabla y(t), \nabla^2 y(t))\right) \, \d x \\
          &\nonumber \quad
          \leq
          \int_\Omega
            (\gamma(\underline{z}_\tau(t)) - \gamma(\underline{z}(t)))
            W(\nabla y(t), \nabla^2 y(t)) \, \d x \\
          &\nonumber \qquad
          +
          \int_\Omega
            \phi \left( (\nabla y)^{-\top}\!(t) \nabla \underline{z}_\tau(t) \right)
            -
            \phi
            \left(
                (\nabla \bar{y}_\tau)^{-\top}\!(t) \nabla \underline{z}_\tau(t)
            \right) \, \d x \\
          &\nonumber \qquad
          -
          \ell(k\tau,y(t),\underline{z}_\tau(t))
          +
          \ell(k\tau,\bar{y}_\tau(t),\underline{z}_\tau(t)),
    \end{align}
    where $k\tau \searrow t$ as $\tau \to 0$.
    Note that the regularizing terms $\mathcal{H}_\tau(\underline{z}_\tau(t))$
    are finite thanks to~the~uniform estimate \eqref{eq:UniformEst},
    and hence have canceled each other.
    Since $\gamma$ and $W$ are lower semi-continuous, by the convexity of $W$ in its last argument and by classical Sobolev embeddings,
    we have lower semicontinuity of the left-hand side of \eqref{eq:improved-z}
    with respect to the convergences \eqref{eq:yWeak} and \eqref{eq:HellyNZ}. Namely,
    \begin{align}
    \label{eq:liminf-en}
        0
        \leq
        \liminf_{\tau \to 0}
        \int_\Omega
          \left(\gamma(\underline{z}_\tau(t)) W(\nabla \bar{y}_\tau(t), \nabla^2 \bar{y}_\tau(t))
          -
          \gamma(\underline{z}(t)) W(\nabla y(t), \nabla^2 y(t))\right)\,\d x;
    \end{align}
    see e.g. \cite{Eise79SLSMSLSC} or  \cite[Cor. 7.9]{FonLeo07MMCV}. 

    Concerning the right-hand side of \eqref{eq:improved-z},
    for every $t\in[0,T]$ we have
    \begin{align*}
        &\int_\Omega
          (\gamma(\underline{z}_\tau(t)) - \gamma(\underline{z}(t)))
          W(\nabla y(t), \nabla^2 y(t))\,\d x \\
        &\quad
        \leq
        \|\gamma(\underline{z}_\tau(t)) - \gamma(\underline{z}(t))\|_\infty
        \int_\Omega
          W(\nabla y(t), \nabla^2 y(t))\,\d x \\
        &\quad
        \leq
        C
        \|\gamma(\underline{z}_\tau(t)) - \gamma(\underline{z}(t))\|_\infty \to 0,
    \end{align*}
    since $\gamma$ is continuous and hence uniformly continuous on $[0,1]$. The second inequality is a consequence of the fact that $y(t)\in \mathcal{Y}$ for almost every $t\in[0,T]$, and of the growth condition \eqref{eq:WCoer}.
    The uniform convergence $\|\underline{z}_\tau(t) - \underline{z}(t)\|_\infty \to 0$
    follows by \eqref{eq:zUnderStrong}, and by the compact Sobolev embedding of $W^{1,\alpha}(\Omega)$ into $C(\bar{\Omega})$ for $\alpha > 3$. 

    For the second term on the right-hand side of \eqref{eq:improved-z},
    we need to exclude concentrations of $\{\nabla \underline{z}_\tau(t)\}$.
    Let us recall that the improved convergence of $\nabla \underline{z}_\tau(t)$
    in \eqref{eq:zUnderStrong} holds only almost everywhere in $(0,T)$,
    while we aim to prove \eqref{eq:ElastConv} for all times.
    Although $\{\nabla \bar{y}_\tau(t)\}$ converges uniformly,
    a blow up
    of $\{\nabla \underline{z}_\tau(t)\}$ may still keep the difference bounded away from zero.
    We hence split this term using the equiintegrable sequence $\{\nabla v_\tau(t)\}$
    provided by Lemma \ref{lemma:Compactness2}, and satisfying \eqref{eq:w1alpha}--\eqref{eq:equiint}.
    On the `bad' set $M_\tau(t)$ we use the growth and coercivity condition
    \eqref{eq:phiCoer}, as well as the bound \eqref{eq:RegEBound}
    in $L^\beta(\Omega)$ and the convergence in \eqref{eq:NyStrong}. By \eqref{eq:meas-Mtau} we obtain
    \begin{align*}
        &\limsup_{\tau \to 0}
        \int_{M_\tau(t)}
            \left|
                \phi
                \left(
                    (\nabla y)^{-\top}\!(t) \nabla \underline{z}_\tau(t)
                \right)
                -
                \phi
                \left(
                    (\nabla \bar{y}_\tau)^{-\top}\!(t) \nabla \underline{z}_\tau(t)
                \right)
            \right| \,
        \d x \\
        &\quad
        \leq
        C
        \limsup_{\tau \to 0}
        \int_{M_\tau(t)}
            |\nabla \underline{z}_\tau(t)|^\alpha \,
        \d x
        \leq
        C
        \limsup_{\tau \to 0}
        \mathcal{L}^3(M_\tau(t))^{1-\frac{\alpha}{\beta}}
        \left(
            \int_\Omega |\nabla \underline{z}_\tau(t)|^\beta
        \right)^{\frac{\alpha}{\beta}}  \\
        &\quad
        \leq
        C
        \limsup_{\tau \to 0}
        |\tau^\alpha|^{1-\frac{\alpha}{\beta}}
        \left( \tau^{-\kappa} \right)^{\frac{\alpha}{\beta}}
        \leq
        C
        \limsup_{\tau \to 0}
        \tau^{\alpha(1-(\alpha+\kappa)/\beta)}
        =
        0,
    \end{align*}
    for $\alpha + \kappa < \beta$.
    On the good set $\Omega\setminus M_{\tau}(t)$, in view of \eqref{eq:phiCoer}, \eqref{eq:NyStrong}, and \eqref{eq:equiint}, the sequence $\{\phi
            \left(
                (\nabla y)^{-\top}\!(t) \nabla \underline{z}_\tau(t)
            \right)
            -
            \phi
            \left(
                (\nabla \bar{y}_\tau)^{-\top}\!(t) \nabla \underline{z}_\tau(t)
            \right)\}$ is equiintegrable. 
    Therefore, since $\LLLL^3(\Omega) < +\infty$ it suffices to show that
    \begin{align}
    \label{eq:claim-meas-conv}
        \left|
            \phi
            \left(
                (\nabla y)^{-\top}\!(t) \nabla \underline{z}_\tau(t)
            \right)
            -
            \phi
            \left(
                (\nabla \bar{y}_\tau)^{-\top}\!(t) \nabla \underline{z}_\tau(t)
            \right)
        \right|
        \to 0
        \quad
        \text{in measure}.
    \end{align}
    Equivalently, we need to show that for every
    $\varepsilon > 0$ and $n \in \NNN$
    there exists $\tau_0 > 0$ such that for all $0<\tau<\tau_0$,
    \begin{align*}
        \LLLL^3(E_{\varepsilon,\tau})
        :=
        \LLLL^3
        \left(
        \left\{
            \left|
            \phi
            \left(
                (\nabla y)^{-\top}\!(t) \nabla \underline{z}_\tau(t)
            \right)
            -
            \phi
            \left(
                (\nabla \bar{y}_\tau)^{-\top}\!(t) \nabla \underline{z}_\tau(t)
            \right)
        \right|
        \geq
        \varepsilon
        \right\}
        \right)
        <
        \frac{1}{n}.
    \end{align*}
    Let then $\varepsilon > 0$ and $n \in \NNN$.
    By the Markov's inequality and the uniform bound
    on $\{\|\nabla \underline{z}_\tau(t)\|_\alpha\}$ given by \eqref{eq:zUnderStrong}, there exists $C_n>0$
    such that for all $\tau > 0$ \ed there holds \EEE
    \begin{align*}
        \LLLL^3(E_\tau)
        :=
        \LLLL^3(\{ |\nabla \underline{z}_\tau(t)| \geq C_n \}) < \frac{1}{n}.
    \end{align*}
    Since $\{\nabla \bar{y}_\tau(t)\}$ is bounded in $L^\infty(\Omega)$ by \eqref{eq:NyStrong},
    we have that in the set $\Omega \setminus E_\tau$
    both $\{(\nabla y)^{-\top}\!(t) \nabla \underline{z}_\tau(t)\}$
    and $\{(\nabla \bar{y}_\tau(t))^{-\top}\!(t) \nabla \underline{z}_\tau(t)\}$
    take value in a bounded, and hence compact, subset $K$ of $\RRR^3$.
    On the one hand, the continuous function $\phi$ is  uniformly continuous on $K$, and there exists $\delta>0$ such that for every $x_1,x_2\in K$ with $|x_1-x_2|<\delta$ there holds $|\phi(x_1)-\phi(x_2)|<\varepsilon$. On the other hand, by 
    \eqref{eq:NyStrong},
    there exists $\tau_0>0$ dependent on $n$ and such that for all $0 < \tau < \tau_0$
    \begin{align*}
        |(\nabla \bar{y}_\tau(t))^{-\top}\!(t) - (\nabla y(t))^{-\top}\!(t)|
        |\nabla \underline{z}_\tau(t)| < \delta
    \end{align*}
    on $\Omega\setminus E_\tau$.
    For these $\tau$ we therefore have $E_{\varepsilon,\tau} \subset E_\tau$
    and hence $\LLLL^3(E_{\varepsilon,\tau}) \leq \LLLL^3(E_\tau) < 1/n$. This in turn yields \eqref{eq:claim-meas-conv}.

    Finally, by assumption the loading satisfies
    \begin{align*}
        &\limsup_{\tau \to 0}
            \left(
            - \ell(k\tau,y(t),\underline{z}_\tau(t))
            + \ell(k\tau,\bar{y}_\tau(t),\underline{z}_\tau(t))
            \right)
        \leq 0.
    \end{align*}

    Altogether, we conclude that
    \begin{align}
    \label{eq:limsup-en}
        \limsup_{\tau \to 0}
        \int_\Omega
          \gamma(\underline{z}_\tau(t)) W(\nabla \bar{y}_\tau(t), \nabla^2 \bar{y}_\tau(t))
          -
          \gamma(\underline{z}(t)) W(\nabla y(t), \nabla^2 y(t)) \,
          \d x
        \leq
        0.
    \end{align}
    Combining \eqref{eq:liminf-en} with \eqref{eq:limsup-en} we deduce \eqref{eq:lim-en}.
    By Lemma \ref{lemma:Visintin}, we infer \eqref{eq:yStrong}.
    The convergence in \eqref{eq:ElastConv} follows by \eqref{eq:yStrong}, Corollary \ref{corollary:WConti}, and by \eqref{eq:zUnderStrong}, which in turn guarantees the uniform convergence
    of $\{\gamma(\bar{z}_\tau(t))\}$ to $\gamma(z(t))$.
\end{proof}

\subsection{Passage to the Limit}
\label{subs:limit}
\ed In this last subsection, relying on the compactness properties established in Subsections \ref{subs:subs}--\ref{subs:better-y}, we show that the limiting pair $(y,z)$ satisfies Definition \ref{def:sep-global}. We begin by proving the \EEE semistability of the limiting pair $(y,z)$. We point out that, in comparison with the previous lemmas, the result holds under slightly more general lower-semicontinuity assumptions on the loading.

\begin{lemma}[Semi-Stability]
	\label{lemma:SemiStab}
    Let the \ed hypotheses \EEE of Lemma \ref{lemma:ImprovedZ} hold,
    and let the loading functional $\ell: \RRR \times \YYYY \times Z \to \RRR$
    be such that
    the map
            $(t,y,z,\tilde{z}) \mapsto \ell(t,y,\tilde{z}) - \ell(t,y,z)$
            is lower semi-continuous on~\ed uniformly bounded \EEE sequences
            in~$\RRR \times \ed \YYYY \EEE \times W^{1,\alpha}(\Omega) \times W^{1,\alpha}(\Omega)$ \ed (see Subsection \ref{subs:state}) \EEE
            that 
            converge
            in~the~strong product topology.
    Then,
    \begin{align}
    \label{eq:semi-staiblity}
        &\forall t \in [0,T], \, \forall \tilde{z} \in Z:
        &&
        \mathcal{E}(t,y(t),z(t))
        \leq
        \mathcal{E}(t,y(t),\tilde{z})
        +
        \mathcal{R}(\tilde{z}-z(t)).
    \end{align}
\end{lemma}
\begin{proof}
    Let $t \in [0,T]$ and $\tilde{z}\in Z$. 
   Let $\tau(t)$ be the $t$-dependent subsequence extracted in Lemma \ref{lemma:ImprovedY}.
    With a slight abuse of notation throughout the proof we omit the explicit $t$-dependence
    and simply write $\tau$.
    As in Lemma \ref{lemma:ImprovedZ}, we cannot use  $\tilde{z} \in Z$ directly as a~competitor
    in~the~discrete semi-stability \eqref{eq:InterpSemiStab},
    because it may not lie in $W^{1,\beta}(\Omega)$. Thus, we work instead with its Lipschitz truncation
     $\tilde{z}_\tau := (\tilde{z})_{\lambda(\tau)}$.
    The dependence of the truncation parameter $\lambda$ on the time step $\tau$
    is again chosen so that $\mathcal{H}_\tau(\tilde{z}_\tau) \to 0$ (see \eqref{eq:blow-up-speed} and \eqref{eq:reg-zero}).
    After moving the loading to the right-hand side, we obtain
    \begin{align*}
        &
        \int_\Omega
          \gamma(\bar{z}_\tau(t))
          W(\nabla \bar{y}_\tau(t), \nabla^2 \bar{y}_\tau(t))
          +
          \phi \left((\nabla \bar{y}_\tau)^{-\top}(t) \nabla \bar{z}_\tau(t) \right) \,
        \d x \\
        &\quad
        +
        \mathcal{H}_\tau(\bar{z}_\tau(t)) \\
        & \quad
        \leq
        \int_\Omega
          \gamma(\tilde{z}_\tau)
          W(\nabla \bar{y}_\tau(t), \nabla^2 \bar{y}_\tau(t))
          +
          \phi \left((\nabla \bar{y}_\tau)^{-\top}(t) \nabla \tilde{z}_\tau \right) \,
        \d x \\
        &\quad
        -
        \ell(k\tau,\bar{y}_\tau(t),\tilde{z}_\tau)
        +
        \ell(k\tau,\bar{y}_\tau(t),\bar{z}_\tau(t))
        + \RRRR(\tilde{z}_\tau - \bar{z}_\tau(t))
        + \mathcal{H}_\tau(\tilde{z}_\tau).
    \end{align*}
    On the left-hand side, the regularization $\mathcal{H}_\tau(\bar{z}_\tau(t))$ converges to zero owing to Lemma \ref{lemma:ImprovedZ}. The remaining energy terms are lower semi-continuous with respect to \eqref{eq:yWeak}-- \eqref{eq:NyNzWeak}.
    The first two terms on the right-hand side are continuous
    with respect to the strong convergence of deformations and damage variables owing to Lemmas \ref{lemma:ImprovedZ} and \ref{lemma:ImprovedY}, Corollary \ref{corollary:WConti},
    and to standard result on Nemytskii operators.
    Finally, by the lower-semicontinuity of the loading functional
    we have
    \begin{align*}
        &\limsup_{\tau \to 0}
            \left(
            - \ell(k\tau,\bar{y}_\tau(t),\tilde{z}_\tau)
            + \ell(k\tau,\bar{y}_\tau(t),\bar{z}_\tau(t))
            \right)
        \leq
        - \ell(t,y(t),\tilde{z})
        + \ell(t,y(t),z(t)).
    \end{align*}
    This completes the proof of \eqref{eq:semi-staiblity}.
\end{proof}

The proof of the stability condition follows by similar arguments.
\begin{lemma}[Stability]
	\label{lemma:Stab}
    Let the hypotheses of Lemma \ref{lemma:ImprovedY} hold
    and let the loading functional $\ell: \RRR \times \YYYY \times Z \to \RRR$
    be such that for each $\tilde{y} \in \YYYY$ the map
            $(t,y,z) \mapsto \ell(t,\tilde{y},z) - \ell(t,y,z)$
            is lower semi-continuous
            on~$\RRR \times W^{2,p}(\Omega) \times W^{1,\alpha}(\Omega)$
            with respect to~the~strong product topology.
    Then,
    \begin{align}
        \label{eq:stability}&&
        \forae t \in [0,T], \, \forall \tilde{y} \in \YYYY:
        &&
        \mathcal{E}(t,y(t),z(t)) \leq
        \mathcal{E}(t,\tilde{y},z(t))
        &&
    \end{align}
\end{lemma}
\begin{proof}
    In proving \ed that the limiting map satisfies \eqref{eq:stability}, \EEE
    we have to restrict ourselves to points $t \in [0,T]$ for which $z(t) = \underline{z}(t)$
    and the strengthened convergence \eqref{eq:zUnderStrong} holds.
    In what follows, we will directly work with the subsequence $\tau(t)$
    identified in Lemma \ref{lemma:ImprovedY},
    and for ease of notation we will omit the $t$-dependence simply
    denoting it by $\tau$.
    Fix $\tilde{y} \in \YYYY$. By using it as a~competitor
    in~the~discrete stability \eqref{eq:InterpStab}
    we obtain, after moving the loading to the right-hand side,
    \begin{align*}
        &\int_\Omega
          \gamma(\underline{z}_\tau(t))
          W(\nabla \bar{y}_\tau(t), \nabla^2 \bar{y}_\tau(t))
          +
          \phi
          \left(
            (\nabla \bar{y}_\tau)^{-\top}\!(t) \nabla \underline{z}_\tau(t)
          \right) \, \d x \\
          & \quad
          \leq
          \int_\Omega
            \gamma(\underline{z}_\tau(t)) W(\nabla \tilde{y}, \nabla^2 \tilde{y})
            +
            \phi \left( (\nabla \tilde{y})^{-\top} \nabla \underline{z}_\tau(t) \right) \,
          \d x \\
          & \qquad
          -
          \ell(k\tau,\tilde{y},\underline{z}_\tau(t))
          + \ell(k\tau,\bar{y}_\tau(t),\underline{z}_\tau(t)).
    \end{align*}
    \ed Owing to \EEE Lemmas \ref{lemma:ImprovedZ} and \ref{lemma:ImprovedY}, we pass to the limit in the left-hand side by standard lower-semicontinuity arguments.
    By the continuity of $\gamma$, the convergence in \eqref{eq:zUnderStrong}, and the fact that
    $\gamma(\underline{z}_\tau(t))W(\nabla \tilde{y},\nabla^2\tilde{y})\leq (\max_{[0,1]} \gamma ) 
    W(\nabla \tilde{y},\nabla^2 \tilde{y})$, the first term on the right-hand side satisfies
    \begin{align*}
        \int_\Omega
            \gamma(\underline{z}_\tau(t))
            W(\nabla \tilde{y}, \nabla^2 \tilde{y}) \,
        \d x
        \to
        \int_\Omega
            \gamma(z(t))
            W(\nabla \tilde{y}, \nabla^2 \tilde{y}) \,
        \d x.
    \end{align*}
    By the lower-semicontinuity assumptions on the loading with respect to the given convergences, we have,
    up to a countable subset $J \subset [0,T]$
    \begin{align*}
        \limsup_{\tau \to 0}
            \left(
            - \ell(k\tau,\tilde{y},\underline{z}_\tau(t))
            + \ell(k\tau,\bar{y}_\tau(t),\underline{z}_\tau(t))
            \right)
        \leq
        - \ell(t,\tilde{y},z(t))
        + \ell(t,y(t),z(t)).
    \end{align*}
    The passage to the limit in the $\phi$-term follows by the strong convergences \eqref{eq:NyStrong} and
    \eqref{eq:zUnderStrong} and the growth condition \eqref{eq:phiCoer}.
\end{proof}

\begin{remark}[\ed Limiting stability\EEE]
    As opposed to Lemma \ref{lemma:ImprovedZ},
    where the strong convergence of the deformations is proven for all times $t \in [0,T]$,
    here the stability holds only almost everywhere in $(0,T)$.
    The difference between the two proof strategies consists in the fact that
    here we needed to prevent both concentrations and oscillations
    of $\{ \nabla \underline{z}_\tau(t) \}$,
    while in Lemma \ref{lemma:ImprovedZ} the possible oscillations
    have been suppressed \emph{for every} $t\in [0,T]$ owing to the uniform convergence
    of $\{\nabla \bar{y}_\tau(t)\}$.
    Consequently, we \ed obtain that the limiting pair $(y,z)$ satisfies the semi-stability and energy inequality for all times $t\in [0,T]$, but that the limiting stability condition only holds for almost every $t\in [0,T]$.\EEE
\end{remark}

We proceed by identifying the limiting \ed work of the loading, \EEE under a further
continuity assumption on $\partial \ell$.
\begin{lemma}[Identification of the limiting \ed work of the loading\EEE]
    \label{lemma:RedPow}
    Let the hypotheses of Lemmas \ref{lemma:SemiStab} and \ref{lemma:Stab} hold.
    Let the reduced power $-\partial_t \ell: [0,T] \times \YYYY \times Z \to \RRR$
        be continuous on \ed uniformly bounded \EEE strongly converging sequences
        in~$\RRR \times \ed \YYYY\EEE \times W^{1,\alpha}(\Omega)$
        \ed (see Subsection \ref{subs:state}). \EEE
    Then,
    \begin{align*}
        \underline{\theta} = \bar{\theta}
        \quad
        \text{in } L^1(0,T)\quad\text{and}\quad\forall t \in [0,T]:
        \quad
        \bar{\theta}^{\text{sup}}(t)
        =
        \partial_t \EEEE(t,y(t),z(t)).
    \end{align*}
\end{lemma}
\begin{proof}
    Fix $t\in [0,T]$.
    Let $\tau(t)$ be the $t$-dependent subsequence
    identified in Lemma \ref{lemma:ImprovedY},
    which is, however, for simplicity henceforth denoted by $\tau$.
    Choose $\{k(\tau)\}\subset \mathbb{N}$,
    also just denoted by $\{k\}$ for simplicity,
    such that $t\in ((k-1)\tau,k\tau] \subset [0,T]$ and $k\tau \searrow t$.
    On the one hand,  by \eqref{eq:ThetaPW} we have
    \begin{align*}
        -\partial_t \ell(k\tau,\bar{y}_\tau(t),\bar{z}_\tau(t))
        =
        \partial_t \EEEE(k\tau,\bar{y}_\tau(t),\bar{z}_\tau(t))
        =
        \bar{\theta}_\tau(t)
        \to
        \bar{\theta}^\text{ sup}(t).
    \end{align*}
    On the other hand, by the improved  convergences \eqref{eq:yStrong} and \eqref{eq:zBarStrong}
    and by the  continuity assumptions on the loading,
    \begin{align*}
        -\partial_t \ell(k\tau,\bar{y}_\tau(t),\bar{z}_\tau(t))
        \to
        -\partial_t \ell(t,y(t),z(t))
        =
        \partial_t \EEEE(t,y(t),z(t)).
    \end{align*}
This yields the second part of the statement.
    \ed In order to prove \EEE the equality of $\bar{\theta}$ and $\underline{\theta}$
    we show that $\underline{\theta}_\tau - \bar{\theta}_\tau \rightharpoonup 0$
    weakly in $L^1(0,T)$,
    by modifying the proof strategy in \cite[proof of Thm.\ 8.9]{RoubicekNPDE2013}.
    We first observe that by a density argument it is enough to consider test functions of the form
    \begin{align*}
        \varphi(t) := c\chi_{[k_1\tau_0,k_2\tau_0]}(t),
    \end{align*}
    where $\tau_0$ is some fixed time step from the discretization,
    $k_1,k_2 \in \NNN$ are such that $0 < k_1\tau_0 < k_2 \tau_0 < T$,
    and $c \in \RRR$.

    Thanks to the equiintegrability of $\bar{\theta}_\tau$ and $\underline{\theta}_\tau$ (see Lemma \ref{lemma:Compactness1}),
    we may suppose without loss of generality that $\frac{\tau_0}{\tau}$ is an integer number.
    Then for all $\tau<\tau_0$ we have
    \begin{align}
        \nonumber
        \int_0^T (\underline{\theta}_\tau(t) - \bar{\theta}_\tau(t)) \varphi(t) \, \d t
        &=
        \int_{k_1\tau_0}^{k_2\tau_0}
            (\underline{\theta}_\tau(t) - \bar{\theta}_\tau(t)) c \,
        \d t \\ \nonumber
        &=
        c
        \sum_{k = k_1\frac{\tau_0}{\tau} + 1}^{k_2\frac{\tau_0}{\tau}}
        \int_{(k-1)\tau}^{k\tau}
            \partial_t \EEEE(t,y^{k-1}_\tau,z^{k-1}_\tau)
            -
            \partial_t \EEEE(t,y^k_\tau,z^k_\tau) \,
        \d t \\ \nonumber
        &=
        c
        \sum_{k = k_1\frac{\tau_0}{\tau} + 1}^{k_2\frac{\tau_0}{\tau}}
        \int_{(k-1)\tau}^{k\tau}
            \partial_t \EEEE(t,y^{k-1}_\tau,z^{k-1}_\tau)
            -
            \partial_t \EEEE(t+\tau,y^k_\tau,z^k_\tau) \,
        \d t \\ \nonumber
        &\quad
        +
        c
        \sum_{k = k_1\frac{\tau_0}{\tau} + 1}^{k_2\frac{\tau_0}{\tau}}
        \int_{(k-1)\tau}^{k\tau}
            \partial_t \EEEE(t+\tau,y^{k}_\tau,z^{k}_\tau)
            -
            \partial_t \EEEE(t,y^{k}_\tau,z^{k}_\tau) \,
        \d t \\ \nonumber
        &=
        c
        \sum_{k = k_1\frac{\tau_0}{\tau} + 1}^{k_2\frac{\tau_0}{\tau}}
        \int_{(k-1)\tau}^{k\tau}
            \partial_t \EEEE(t,\underline{y}_\tau(t),\underline{z}_\tau(t))
            -
            \partial_t \EEEE(t+\tau,\underline{y}_\tau(t+\tau),\underline{z}_\tau(t+\tau)) \,
        \d t \\
        \label{eq:int4.6c}
        &\quad
        +
        c
        \int_{k_1\tau_0}^{k_2\tau_0}
            \partial_t \EEEE(t+\tau,\bar{y}_\tau(t),\bar{z}_\tau(t))
            -
            \partial_t \EEEE(t,\bar{y}_\tau(t),\bar{z}_\tau(t)) \,
        \d t.
    \end{align}
    For the first term on the right-hand side of \eqref{eq:int4.6c}, we infer that
    \begin{align*}
        &\sum_{k = k_1\frac{\tau_0}{\tau} + 1}^{k_2\frac{\tau_0}{\tau}}
        \int_{(k-1)\tau}^{k\tau}
            \partial_t \EEEE(t,\underline{y}_\tau(t),\underline{z}_\tau(t))
            -
            \partial_t \EEEE(t+\tau,\underline{y}_\tau(t+\tau),\underline{z}_\tau(t+\tau)) \,
        \d t \\
        &=
        \sum_{k = k_1\frac{\tau_0}{\tau}+1}^{k_2\frac{\tau_0}{\tau}}
        \int_{(k-1)\tau}^{k\tau}
            \partial_t \EEEE(t,\underline{y}_\tau(t),\underline{z}_\tau(t)) \,
        \d t
        -
        \sum_{k = k_1\frac{\tau_0}{\tau} + 2}^{k_2\frac{\tau_0}{\tau}+1}
        \int_{(k-1)\tau}^{k\tau}
            \partial_t \EEEE(t,\underline{y}_\tau(t),\underline{z}_\tau(t)) \,
        \d t \\
        &=
        \int_{k_1\tau_0}^{k_1\tau_0+\tau}
            \underline{\theta}_\tau(t) \,
        \d t
        -
        \int_{k_2\tau_0}^{k_2\tau_0+\tau}
            \underline{\theta}_\tau(t) \,
        \d t.
    \end{align*}
    Owing to the equiintegrability of~$\underline{\theta}_\tau$
    these two integrals can be made arbitrarily small as $\tau \to 0$.
    Eventually, the second term on the
    right-hand side of \eqref{eq:int4.6c} converges to zero by the Dominated Convergence Theorem.
    Indeed, the pointwise convergence follows by the strong continuity of $\partial_t \ell$, as well as by the  convergences \eqref{eq:yStrong} and \eqref{eq:zBarStrong},
    and by the uniform bound of $(\det \nabla \bar{y}_\tau(t))^{-1}$ in $L^s(\Omega)$ provided by the growth conditions in \eqref{eq:WCoer} and by \eqref{eq:ElastConv}.
    The integrable majorant is obtained arguing exactly as in Step 1 of the proof
    of Lemma \ref{lemma:Compactness1}.
\end{proof}

The next lemma shows that the limiting pair $(y,z)$ satisfies an energy inequality. We point out that, in proving the (semi)-stability condition, lower
semicontinuity of the loading with respect to the weak topology was needed, whereas
for obtaining the energy inequality  we need to enforce continuity  with respect to the strong product topology.

\begin{lemma}[Energy Inequality]
    \label{lemma:EnIneq}
    Let the hypotheses of Lemma \ref{lemma:RedPow} hold
    and let the loading $\ell:[0,T] \times \YYYY \times Z \to \RRR$
    be continuous on \ed uniformly bounded \EEE
    sequences
            in~$\RRR \times \ed \mathcal{Y}\EEE\times W^{1,\alpha}(\Omega)$
    that converge in the strong product topology \ed (recall \eqref{eq:def-Y} and Subsection \ref{subs:state}).\EEE 
    Then,
    \begin{align*}
      &\forall t_1, t_2 \in I, \, t_1 < t_2:
      &
      \mathcal{E}(t_2,y(t_2),z(t_2))
      +
      \diss_\mathcal{R}(z;[t_1,t_2]) 
      \leq
      \mathcal{E}(t_1,y(t_1),z(t_1))
      +
      \int_{t_1}^{t_2}
        \partial_t \EEEE (t,y(t),z(t)) \, \d t.
    \end{align*}
\end{lemma}
\begin{proof}
    Fix $t_1,t_2\in [0,T]$ with $t_1<t_2$,
    and let $\tau(t_1)$ be the subsequence from Lemma \ref{lemma:Compactness2}
    for which the convergences \eqref{eq:ElastConv} and \eqref{eq:yStrong}
    for \ed $\{\bar{y}_{\tau(t_1)}(t_1)\}$ \EEE hold at time $t_1$.
    \ed By the uniform bound in \eqref{eq:new-bound} and by the coercivity assumptions in the statement of the lemma, we extract a~further subsequence $\tau'(t_1)$ such that
    we have also \EEE
    \begin{align}
        \label{eq:y2Weak}
        \bar{y}_{\tau'(t_1)}(t_2) &\rightharpoonup \xi(t_2)
        &&\text{ in } W^{2,p}(\Omega).
        \end{align}
       \ed Arguing as in Lemma \ref{lemma:Compactness1}, we can assume that we also have \EEE \begin{align}\label{eq:Ny2NzStrong}
        (\nabla \bar{y}_{\tau'(t_1)})^{-\top}\!(t_2) \nabla \bar{z}_{\tau'(t_1)}(t_2)
        & \to (\nabla \xi)^{-\top}\!(t_2) \nabla z(t_2)
        &&\text{ in } L^\alpha(\Omega).
    \end{align}
    
    For ease of notation we omit the explicit dependence of this subsequence
    on $t_1$ and
    denote it by $\tau'$
    in the following.
    Let also $\{k_1({\tau'})\},\, \{k_2({\tau'})\}\subset \mathbb{N}$,
    also denoted simply by $\{k_1\}$ and $\{k_2\}$,
    be such that $t_1 \in ((k_1-1){\tau'},k_1{\tau'}] \subset [0,T]$
    and $t_2 \in ((k_2-1){\tau'},k_2{\tau'}] \subset [0,T]$
    with $k_1{\tau'} \searrow t_1$ and $k_2{\tau'} \searrow t_2$.
    Using the definition of the discrete dissipation, the
    absolute continuity of the energy \eqref{eq:E2},
    and the discrete energy inequality \eqref{eq:DiscreteEIneq},
    we obtain
    \begin{align}
    \label{eq:provisional-en}
        &\EEEE(t_2,\bar{y}_{\tau'}(t_2),\bar{z}_{\tau'}(t_2))
        +
        \diss_{\RRRR}(\bar{z}_{\tau'};[t_1,t_2]) \\
        \nonumber&=
        \EEEE(k_2{\tau'},\bar{y}_{\tau'}(k_2{\tau'}),\bar{z}_{\tau'}(k_2{\tau'}))
        +
        \diss_{\RRRR}(\bar{z}_{\tau'};[k_1{\tau'},k_2{\tau'}])
        -
        \int_{t_2}^{k_2{\tau'}} \bar{\theta}_{\tau'}(t) \, \d t \\
        \nonumber&\leq
        \EEEE(k_1{\tau'},\bar{y}_{\tau'}(k_1{\tau'}),\bar{z}_{\tau'}(k_1{\tau'}))
        +
        \int_{k_1{\tau'}}^{k_2{\tau'}} \underline{\theta}_{\tau'}(t) \, \d t
        -
        \int_{t_2}^{k_2{\tau'}} \bar{\theta}_{\tau'}(t) \, \d t +\HHHH_{\tau'} (\bar{z}_{\tau'}(k_1{\tau'}))-\HHHH_{\tau'}(\bar{z}_{\tau'}(k_2{\tau'}))\\
        \nonumber&=
        \EEEE(t_1,\bar{y}_{\tau'}(t_1),\bar{z}_{\tau'}(t_1))
        +
        \int_{t_1}^{t_2} \underline{\theta}_{\tau'}(t) \, \d t
        -
        \int_{t_1}^{k_1{\tau'}}\underline{\theta}_{\tau'}(t) \, \d t
        -
        \int_{t_2}^{k_2{\tau'}}
            \bar{\theta}_{\tau'}(t) - \underline{\theta}_{\tau'}(t) \, \d t \\
        \nonumber& \quad
        + \HHHH_{\tau'} (\bar{z}_{\tau'}(k_1{\tau'}))
        - \HHHH_{\tau'}(\bar{z}_{\tau'}(k_2{\tau'})).
    \end{align}

    For the first term on the right-hand side of \eqref{eq:provisional-en} we apply
    the convergence of elastic energies in \eqref{eq:ElastConv},
    as well as the continuity of the phase field energy $\phi$ and of the loading $\ell$
    with respect to the strong convergence \eqref{eq:yStrong}
    and \eqref{eq:zBarStrong} of $\bar{y}_{\tau'}(t_1)$ and $\bar{z}_{\tau'}(t_1)$.
    For the second term we exploit the convergence of the \ed work of the loading \EEE
    \eqref{eq:ThetaWeak}. The third and fourth energy terms converge to zero \ed due \EEE to the equiintegrability
    of $\{\bar{\theta}_{\tau'}\}$ and $\{\underline{\theta}_{\tau'}\}$.
    Eventually, the regularization terms converge to zero owing to \eqref{eq:reg-zero}.
    
    On the left-hand side of \eqref{eq:provisional-en} we exploit the lower semicontinuity of the energy $\EEEE$
    with respect to the strong convergences \eqref{eq:zBarStrong}
    and \eqref{eq:Ny2NzStrong},
    and to the weak convergence \eqref{eq:y2Weak}.
    The lower-semicontinuity of the dissipation $\diss_\RRRR(\bar{z}_{\tau'};[0,t])$ is a consequence of \eqref{eq:HellyDiss} and \eqref{eq:HellyDissLSC}.
    Thus, using in addition the stability condition \eqref{eq:stability},
    we obtain
    \begin{align*}
        \EEEE(t_2,y(t_2),z(t_2))
        +
        \diss_{\RRRR}(z;[t_1,t_2])
        &\leq
        \EEEE(t_2,\xi(t_2),z(t_2))
        +
        \diss_{\RRRR}(z;[t_1,t_2]) \\
        &\leq
        \EEEE(t_1,y(t_1),z(t_1))
        +
        \int_{t_1}^{t_2} \underline{\theta}(t) \, \d t.
    \end{align*}
    To conclude we recall that by Lemma \ref{lemma:RedPow}
    we have $\underline{\theta} = \bar{\theta} \leq \bar{\theta}^\text{ sup}$
    a.e. in $(0,T)$
    and that $\bar{\theta}^\text{ sup}(t) = \partial_t \EEEE(t,y(t),z(t))$
    for every $t\in [0,T]$.
\end{proof}

\begin{remark}[\ed Limiting energy inequality \EEE]
    When analyzing the proof we see that \textbf{energy equality} is bound to break
    not only when the weak limit
    of $\partial_t \EEEE(t,\bar{y}_\tau(t),\bar{z}_\tau(t))$
    is strictly less than the maximal power $\partial_t \EEEE(t,y(t),z(t))$,
    but also whenever $\EEEE(t_2,y(t_2),z(t_2)) < \EEEE(t_2,\xi(t_2),z(t_2))$,
    i.e. all the cluster points of $\bar{y}_{\tau}(t_2)$
    maximizing the \ed work of the loading \EEE at time $t_2$
    have strictly less energy than any limit $\xi(t_2)$
    which maximizes the reduced power at time $t_1$.
\end{remark}

We are finally in a position to prove our main result.
\begin{proof}[Proof of Theorem \ref{thm:existence}]
    Most of the statement has been proven in Lemmas
    \ref{lemma:Discrete}--\ref{lemma:ex-2-grad}
    and \ref{lemma:Compactness1}--\ref{lemma:EnIneq}.
    The injectivity of the limiting deformation $y$ follows directly from the definition of the space $\YYYY$
    in \eqref{eq:def-Y}, the Ciarlet--Nečas condition \eqref{eq:C-N},
    and the integrability of the distortion coefficient \eqref{eq:H-K},
    which in turn is a consequence of Hölder's inequality.
    Due to the regularization $\HHHH_\tau$
    we also have to show that the limiting damage variable $z$ satisfies the initial condition.
    This can be inferred by the equality
    $\bar{z}_\tau(0) = \underline{z}_\tau(0) = z^0_\tau = (z^0)_{\lambda(\tau)}$,
     and by the convergences
    \begin{align*}
        \bar{z}_\tau(0) \rightharpoonup^* z(0) \quad \text{in } L^\infty(\Omega),
    \end{align*}
     and 
    \begin{align*}
        (z^0)_{\lambda(\tau)} \to z^0 \quad \text{in } W^{1,\alpha}(\Omega),
    \end{align*}
    being a consequence \ed of \EEE \eqref{eq:HellyZ} and of  Lemma \ref{lemma:Truncation}, respectively.
\end{proof}

\section{Discussion of the results}
\label{sec:disc}
We generalized the concept of local solutions,
developed in \cite{Toader-Zanini,Roubicek2015}
 for separately convex energies 
by introducing the notion of separately global solutions in Definition \ref{def:sep-global}. \ed This novel class of local solution allows to encompass in the analysis energy densities which are non-convex with respect to the deformation variable and where convexity is only enforced with respect to higher-order derivatives. Large-strain formulations including injectivity constraints and blow up for extreme compressions, as well as mixed Eulerian-Lagrangian energetic contributions are also incorporated within the framework described in this paper.
Unlike energetic solutions, separately global solutions do not force `too early jumps'
and have therefore a broad range of application. 

Our main result in Theorem \ref{thm:existence} is the existence of separately global solutions in the setting of gradient bulk damage and for non-simple materials. 

The mathematical regularization provided by the coercivity and convexity of the energy with respect to higher-order derivatives of the deformations is essential to handle nonlinear Eulerian-Lagrangian couplings between the deformations and the internal variables. This higher-order formulation is indeed necessary for passing from time-discrete interpolants to time-continuous solutions.
It is, on the other hand, not needed for proving existence of time-discrete solutions, for which
polyconvex energies dependent only on the gradient of the deformations can be included in the analysis.
This is a consequence of the alternating minimization scheme \eqref{eq:DiscreteSystem}
which makes the discrete problem effectively decoupled.

A further key point of our analysis is the fact that,
at the time-discrete level, different piecewise-constant interpolants of the damage variable come into play in the discrete stability \eqref{eq:InterpStab} and semi-stability \eqref{eq:InterpSemiStab}
conditions. To pass to the time-continuous setting in the semi-stability condition and to obtain the limiting energy inequality, on the other hand, 
the selection of suitable time-dependent subsequences of the piecewise-constant interpolants is needed in order to guarantee strong convergence of deformations and internal variables in the appropriate topologies. This in turn, is related to the fact that in the energy inequality in Lemma \ref{lemma:EnIneq}, both the left- and the right-hand side of the estimate involve energy contributions evaluated at arbitrary times $t\in [0,T]$, so that
on the left-hand side mere lower semicontinuity does not suffice.

For this reason, the selection of time-dependent subsequences in Lemmas \ref{lemma:ImprovedY} and \ref{lemma:ImprovedZ} is hinged upon novel techniques,  
not being simply based on error estimates between the right- and left-continuous piecewise-constant
interpolants. For proving the energy inequality,
apart from selecting deformations that maximize the work of the loading
in the sense of \eqref{eq:ThetaPW},
we also strongly rely on the global minimality of the elastic variable, cf. \eqref{eq:stability}.
These issues are connected to the non-uniqueness of deformations caused by the lack of convexity of the energy, 
which makes the selection of $t$-dependent subsequences in Lemma \ref{lemma:Compactness2}
quite subtle (see also \cite{FraMie06ERCR}).\EEE

\ed Although the ansatz for the stored energy density in \eqref{eq:Energy}
is rather specific,
the analysis is prone to be extended to different mathematical models in which the highest order terms for the elastic and internal variables are decoupled.
Possible forthcoming generalizations include \EEE e.g. phase transitions and finite plasticity.
In this latter setting, the convexity in the internal variable would be lost
and hence the connection with the local stability \eqref{localstability} broken.
Nevertheless the existence theory based on global semi-stability,
replacing the dissipation potential with a~dissipation distance,
would be in principle feasible.

\ed A further generalization concerns
the regularization of the energy by the hessian of the deformations. Working with the nonlinear coupling  it would be natural to replace the coercivity requirement with respect to the full hessian with the notion of gradient polyconvexity (GPC), see \cite{BeKrSc17NLMGP}.
On the one hand, in fact, \EEE GPC allows for the existence of minimizers under weaker regularity
assumptions on the deformations.
\ed On the other hand, \EEE when the boundedness of $\det \nabla y$ from zero is assumed 
(which was \ed a \EEE key ingredient in our analysis),
GPC still guarantees integrability of the second gradient.
It remains an open question, whether the existence theory
\ed can be fully extended to the GPC setting. \EEE

%
%
\appendix
\section{Analytical Tools}
\label{app:anal}

\begin{proof}[Proof of Corollary \ref{corollary:HK}]
   As $p>3$,  $y\in\YYYY$ can be considered  continuous. Assume without loss of generality that $0\in\Omega$, and assume by contradiction that $\det\nabla y(0)=0$.  Then we have for every $x\in\Omega$ that 
   \begin{align}
       |\det \nabla y (x)|
       \le
       C(1+|\nabla y(x)|^2 +|\nabla y(0)|^2)|\nabla y(x)-\nabla y(0)|
       \le \tilde{C} |x|^\alpha\ , 
   \end{align}
   where $\alpha= 1-3/p$ is given by the continuous embedding
   $W^{2,p}(\Omega) \hookrightarrow C^{0,\alpha}(\Omega)$, and where the first inequality follows from the local Lipschitz property of $F\mapsto\det F$, see also \cite[Prop.~2.32]{Daco89DMCV}.
   Hence, 
   \begin{align*}
        \int_{\Omega}
            \frac{{\rm d} x}{(\det\nabla y)^s} \, \d x
        &\ge
        \int_{B(0,r)} \frac{{\rm d} x}{(\det\nabla y)^s} \, \d x
        \ge
        \int_{B(0,r)}\frac{{\rm d} x}{\tilde{C}^s|x|^{\alpha s}}\\
        &\ge
        \frac{4}{3}\pi r^3\frac{1}{\tilde{C}^s r^{\alpha s}}\ .
   \end{align*}
   However, the last expression  diverges for $r\to 0$ if $3-\alpha s<0$, i.e.,
   if $s>3p/(p-3)$. If $s=3p/(p-3)$ then we have that for every $r>0$ small enough
   \begin{align*}
    \int_{B(0,r)}\frac{{\rm d} x}{(\det\nabla y)^s}\ge\frac{4\pi}{3\tilde{C}^s}\ ,
   \end{align*}
   which is not possible if $(\det\nabla y)^{-1}\in L^s(\Omega)$,
   because the Lebesgue integral is absolutely continuous. 
   Altogether, we proved that  $y\not\in \YYYY$.
   If $0\in\partial\Omega$ then we proceed in a similar way.
   Namely, as  $\Omega$ is Lipschitz, it has the cone property.
   This implies that $\mathcal{L}^3((B(0,r)\cap\Omega)\ge \tilde C r^3$
   for $r>0$ small and some $\tilde C>0$ independent of $r$.
   Hence $\det\nabla y>0$ in $\bar\Omega$ and it is a continuous function.
   Inevitably, it is bounded from below by a positive constant in $\bar\Omega$.

    For the proof of the uniform bound we refer
    to~\cite[Proof of Prop. 5.1]{BeKrSc17NLMGP},
    from which we can apply the procedure for an infinite number
    of minimizers.
\end{proof}
\ed We proceed by providing a proof of Lemma \ref{lemma:Visintin}. \EEE

\begin{proof}
    Consider the sequence $\{(z_k,\nabla y_k, \nabla^2 y_k)\}$.
    By the fundamental theorem on Young measures,
    see \cite[Thm. 6.2]{Pedregal1997} or \cite{Ball81GISF},
    there is a subsequence generating a Young measure $\nu$.
    Thanks to \cite[Prop.~6.13]{Pedregal1997} the strong convergence
    of $z_\tau$ and $\nabla y_\tau$ implies
    \begin{align*}
        \nu_x=\delta_{z(x)} \otimes \delta_{\nabla y(x)} \otimes \mu_x
        \quad
        \text{for a.e. } x \in \Omega
    \end{align*}
    where $\mu={\mu_x}$ is the Young measure generated by $(\nabla^2 y_k)_{k\in\NNN}$.
    In view of \eqref{limit11} and
    \cite[Thm.~6.11]{Pedregal1997} applied to $\tilde h:=h+c$ we get 
    \begin{align*}
        \int_\Omega \tilde h(x,z(x),\nabla y(x),\nabla^2 y(x)) \, \d x
        &=
        \liminf_{k\to\infty}
            \int_\Omega \tilde h(x,z_k(x),\nabla y_k(x),\nabla^2y_k(x)) \, \d x \\
        &\ge
        \int_\Omega \int_{\RRR^{3\times 3\times 3}}
            \tilde h(x,z(x),\nabla y(x),G)\mu_x(\d G) \, \d x\ .
    \end{align*}
    The strict convexity of $\tilde h$ in in the last variable together with Jensen's inequality imply that
    $\mu_x=\delta_{\nabla^2y(x)}$ for almost every $x\in\Omega$.
    This further implies that $\nabla ^2 y_k$ converges to $\nabla ^2 y$ in measure; see \cite[Cor.3.2]{Mueller1999}.
    Finally, \eqref{limit11} with \cite[Corollary 6.10]{Pedregal1997}
    yields the strong convergence of $y_k \to y$ in $W^{1,p}(\Omega;\RRR^3)$.
\end{proof}

\begin{remark}
    For a reader not familiar with Young measures we sketch a different proof,
    which, however, still relies on the coercivity and growth of the energy,
    and on its equiintegrability,
    and which resembles in many aspects the original proof.
    The key idea is to use the classical result by A. Visintin \cite{Visintin1984} to infer that
    \begin{align*}
        \int_\Omega h(x,z(x),\nabla y(x),\nabla^2 y_k(x))\, \d x
        \to
        \int_\Omega h(x,z(x),\nabla y(x),\nabla^2 y(x)) \, \d x.
    \end{align*}
    To prove the thesis we then need to show that
    \begin{align*}
        \int_\Omega
            \left(h(x,z_k(x),\nabla y_k(x),\nabla^2 y_k(x))
            -
            h(x,z(x),\nabla y(x),\nabla^2 y_k(x))\right) \,
        \d x
        \to 0.
    \end{align*}
    Proving that the integrand converges in measure is standard, it thus remains only to prove the equiintegrability of $h(x,z_k(x),\nabla y_k(x),\nabla^2 y_k(x))$.
    The latter follows by the convergence \eqref{limit11}.
    Indeed, the lower semicontinuity of the functional implies convergence
    on all measurable subset of $\Omega$.
    Hence the limit of the integrands in the sense of the biting convergence coincides with $h(x,z(x),\nabla y(x),\nabla^2 y(x))$ (see \cite[Theorem 6.6]{Pedregal1997}).
    Since \eqref{limit11} holds, we have by \cite[Lemma 6.9]{Pedregal1997} that
    the integrands converge in the weak $L^1$-topology. This, in turn, yields the equiintegrability of the sequence.
\end{remark}

\ed We conclude this appendix with the statements of two technical lemmas which have been instrumental for the proof of our main result. \EEE

\begin{lemma}[Decomposition Lemma; \cite{FoMuPe98ACOEGG}]
    \label{lemma:Decomposition}
    Let $\Omega \subset \RRR^3$ be a Lipschitz domain
    and let $\{w_i\} \subset W^{1,\alpha}(\Omega)$, $\alpha > 1$, be bounded.
    Then there exists a subsequence $\{w_j\}$ and a sequence
    $\{v_j\} \subset W^{1,\alpha}(\Omega)$ such that
    \begin{align*}
        \|v_j\|_{1,\infty} &\leq C(\alpha,\Omega) j, \\
        \LLLL^3 (M_j) &\leq \frac{C}{j^\alpha}, \\
        \{|\nabla v_j|^\alpha\} & \text{ is equiintegrable},
    \end{align*}
    where
    $M_j := \{ x \in \Omega : w_j(x) \neq v_j(x)
        \text{ or } \nabla w_j(x) \neq \nabla v_j(x) \}$.
\end{lemma}

\begin{lemma}[Lipschitz Truncation; \cite{Ziemer1989}]
    \label{lemma:Truncation}
    Let $\Omega \subset \RRR^{3}$ be a Lipschitz domain,
    $1 < \alpha < +\infty$,
    and $u \in W^{1,\alpha}(\Omega)$.
    Then for every $\lambda > 0$ there exists $u^\lambda \in W^{1,\infty}(\Omega)$
    such that
    \begin{align*}
        \| u^\lambda \|_{1,\infty} &\leq C(\alpha,\Omega) \lambda, \\
        \LLLL^3(\{u \neq u^\lambda \text{ or } \nabla u \neq \nabla u^\lambda\})
        &\leq
        \frac{C(\alpha)}{\lambda^\alpha}
        \int_{|\nabla u| \geq \lambda/2} |\nabla u(x)|^\alpha \, \d x \\
        \| u^\lambda \|_{1,\alpha} &\leq C(\alpha,\Omega) \| u \|_{1,\alpha}.
    \end{align*}
    In particular, denoting
    $M_\lambda
        :=
        \{ x \in \Omega : u(x) \neq u^\lambda(x)
                        \text{ or }
                        \nabla u(x) \neq \nabla u^\lambda(x) \}$,
    we have
    \begin{align*}
        \lim_{\lambda \to \infty} \lambda^\alpha \LLLL^3(M_\lambda) = 0,
        \quad
        \text{and}
        \quad
        \lim_{\lambda \to \infty} \|u^\lambda - u\|_{1,\alpha} = 0.
    \end{align*}
\end{lemma}

\section*{acknowledgements}
We acknowledge support
from the Austrian Science Fund (FWF) projects F65 and V 662,
from the FWF-GA\v{C}R project I 4052/19-29646L, 
from the OeAD-WTZ project CZ04/2019 (M\v{S}MT\v{C}R 8J19AT013),
from the GAUK Project No. 670218,
and from Charles University Research Program No. UNCE/SCI/023.

%
%
\bibliographystyle{plain}
\def\cprime{$'$} \def\ocirc#1{\ifmmode\setbox0=\hbox{$#1$}\dimen0=\ht0
  \advance\dimen0 by1pt\rlap{\hbox to\wd0{\hss\raise\dimen0
  \hbox{\hskip.2em$\scriptscriptstyle\circ$}\hss}}#1\else {\accent"17 #1}\fi}

\end{document}